\title{Nonparametric methods controlling the median of the false discovery proportion} 

\author{Jesse Hemerik\footnote{Econometric Institute, Erasmus University Rotterdam, The Netherlands. e-mail: hemerik@ese.eur.nl}
}

\documentclass[11pt]{article}

\usepackage{amsmath}
\usepackage{amsfonts}
\usepackage{bm}
\usepackage{bbm}
\usepackage{amsthm}
\usepackage{comment}
\usepackage[nottoc]{tocbibind}
\usepackage{graphicx}

\usepackage{subcaption}

\usepackage{changepage}
\usepackage{natbib}
\usepackage{abstract}
\usepackage[colorlinks=true,citecolor=blue,runcolor=blue,linkcolor=blue, pdfborder={0 0 0}]{hyperref}

\usepackage{tikz}
\usepackage{pgfplots}

\usepackage{algorithm}
\usepackage[noend]{algorithmic}

\theoremstyle{plain}
\newtheorem{theorem}{Theorem}[section]

\newtheorem{lemma}[theorem]{Lemma}
\newtheorem{proposition}[theorem]{Proposition}

\newtheorem*{example}{Example}

\theoremstyle{definition}

\newtheorem{assumption}{Assumption}
\newtheorem{remark}{Remark}

\usepackage[margin=3cm]{geometry}

\newcommand{\N}{\mathcal{N}}
\newcommand{\R}{\mathcal{R}}
\newcommand{\M}{\mathcal{M}}
\newcommand{\G}{\mathcal{G}}
\newcommand{\X}{\mathcal{X}}
\newcommand{\C}{\mathcal{C}}

\newcommand{\co}{[c,\infty)}
\newcommand{\cop}{\mathbb{T}}
\newcommand{\pr}{\mathbb{P}}
\newcommand{\de}{\delta}
\newcommand{\lt}{\tilde{\phi}}  

\newcommand{\pv}{p\text{-value}} 
\newcommand{\pvs}{p\text{-values}}

\newcommand{\black}{\color{black}}

\newcommand{\citp}{\citep}
\newcommand{\citt}{\citet}

\begin{document}
\maketitle

\begin{abstract}
\noindent  When testing many hypotheses, often we do not have strong expectations about the directions of the effects. In some situations however, the alternative hypotheses are that the parameters lie in a certain direction or interval, and it is in fact expected that most hypotheses are false. This is often the case when researchers perform multiple noninferiority or equivalence tests, e.g. when testing  food safety with metabolite data. The goal is then to use data to corroborate the expectation that most hypotheses are false. We propose a nonparametric multiple testing approach that is powerful in such situations. If the user's expectations are wrong, our approach will still be valid but have low power. Of course all multiple testing methods become more powerful when appropriate one-sided instead of two-sided tests are used, but our  approach often has superior power then. 
The proposed methods are not at all limited to safety testing and can be used for testing hypotheses about various kinds of parameters, such as coefficients of a model.
The methods in this paper control the median of the false discovery proportion (FDP), which is the fraction of false discoveries among the rejected hypotheses. This approach is comparable to false discovery rate control, where one ensures that the mean rather than the median of the FDP is small. Our procedures make use of a symmetry property of the test statistics, do not require independence and have finite-sample properties.
\\
\\
\emph{keywords:} Equivalence testing, false discovery proportion, non-inferiority testing, nonparametric, symmetry.
\end{abstract}

\section{Introduction}
In many settings where multiple hypotheses are tested, it is expected that most null hypotheses are true or approximately true.
However, there are also applications when one expects a priori that most of the parameters lie in a particular direction or lie within some interval, and the goal is to corroborate this expectation \citp{berger1996bioequivalence,hasler2013simultaneous,kang2014statistical,engel2021equivalence,leday2022multivariate}.  An illustration is provided in Figure \ref{fig:scatter2}.
 For example, before  a new crop---e.g. one that has been genetically modified---is introduced in the European Union,  the European Food Safety Authority often requires analyzing the concentrations of various molecules in the crop \citp{leday2023improved}. It must then for example be shown that the mean concentrations of certain analytes fall below a certain threshold or within some equivalence interval.
In the latter case, this means that one tests null hypotheses of non-equivalence. In some cases researchers test one global null hypothesis about all the parameters \citp{wang1999statistical,chervoneva2007multivariate,hoffelder2015multivariate}, while in other cases the problem is treated as a multiple testing problem 
 \citp{quan2001assessment,logan2008superiority,hua2016multiplicity}, especially when there are many parameters \citp{vahl2017statistical, leday2023improved}.
The null hypotheses are then often of the form $H_0:|\theta| \geq \delta$ and the alternatives  are $H_a:|\theta| < \delta$, with $\theta\in \mathbb{R}$ a parameter and $\delta>0$.
Researchers then often expect a priori that most or all hypotheses are false, and the goal is to statistically corroborate this using some multiple testing method \citp{vahl2017statistical, leday2023improved}. Often it is not possible to reject all hypotheses---even with the intersection-union principle \citp{berger1996bioequivalence, berger1997likelihood,hoffelder2015multivariate}---but typically the data suggest that most hypotheses are false. There is a strong need for innovative multiple testing methods that are tailored to such data.

\begin{figure}[ht] 
\centering
  \includegraphics[width=\linewidth]{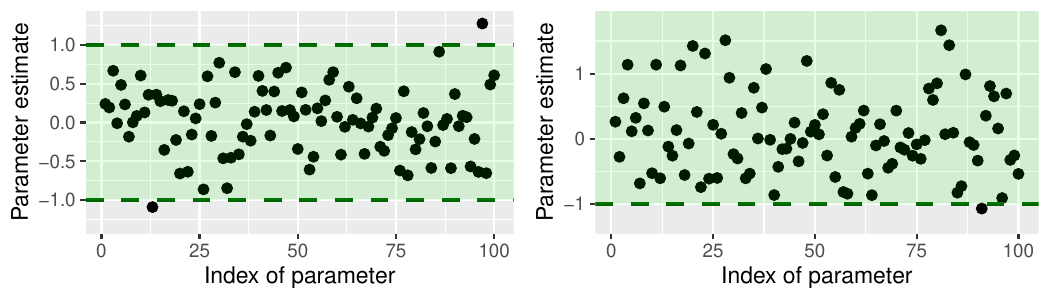}
  \caption{\emph{In some situations,  e.g. food safety testing, one wants to show that many  parameters lie in a certain interval or in a particular direction. In the figure on the left and on the right, the parameter estimates suggest that most parameters lie within the regions $(-1,1)$ and $(-1,\infty)$ respectively. In many such problems, based on prior knowledge, it is already expected that  most parameters lie within such a  region, and the goal is to corroborate that with data.
This paper proposes powerful, nonparametric (or semiparametric) procedures that do this. In Sections \ref{secnovel}-\ref{secet}, we  select  indices of parameters that seem to fall within the region, and estimate the number of incorrectly selected indices, i.e. the indices of parameters that in fact lie outside the region.
 In Section \ref{secfdx}, we provide methods that select a set of indices such that the median of the fraction of incorrectly selected indices---i.e., the median of the FDP---stays below some chosen small value $\gamma\in[0,1)$. }
} \label{fig:scatter2}
\end{figure} 

In this paper, we propose novel multiple testing methods that are powerful in such situations. The methods are not at all limited to safety testing and can be used for testing hypotheses about various kinds of parameters, such as coefficients in a generalized linear model (an example is in Section \ref{secdata}).
 The novel methods connect to a large literature, which we review first, before discussing this paper's contributions.

\subsection{Existing multiple testing approaches}
When multiple hypotheses are tested, the fraction of false discoveries among all rejected hypotheses is called the \emph{false discovery proportion} (FDP). When there are no rejections, the FDP is defined to be 0.
There exist a few different notions of FDP control. The oldest one is familywise error rate (FWER) control, which means guaranteeing that the FDP is 0  with high confidence \citp{goeman2014multiple}.
Apart from FWER control, the most well-known notion is to control the expected value $\mathbb{E}(FDP)$, which is called the \emph{false discovery rate} (FDR). FDR control means ensuring that $FDR\leq\gamma$, where $\gamma\in(0,1)$ is some user-specified bound \citp{benjamini1995controlling,benjamini2001control,dickhaus2014simultaneous,wang2022false}.

Another manner of providing statements on the FDP  is to fix some rejection region, reject all hypotheses with test statistics in this region and then provide a possibly data-dependent bound $\bar{FDP}$ for the FDP, which satisfies $\mathbb{P}(FDP\leq \bar{FDP} ) \geq 1-\alpha$. Here $\alpha\in (0,1)$ is a fixed, user-defined value, e.g. $0.05$. Permutation-based approaches for obtaining such bounds are provided in \citt{hemerik2018false}.
A related goal is to guarantee that the $FDP$ is small with large probability, i.e. that   $\mathbb{P}(FDP\leq \gamma ) \geq 1-\alpha$, where $\gamma\in[0,1)$ is chosen by the user. This is called \emph{false discovery exceedance control} or FDX control \citp{van2004augmentation,lehmann2005generalizations, romano2006stepdown, romano2006stepup, romano2007control, guo2007generalized, farcomeni2008review, roquain2011type, guo2014further, delattre2015new, javanmard2018online, ditzhaus2019variability, 
dohler2020controlling, basu2021empirical, miecznikowski2023exceedance}. FDX controlling methods  vary in the strictness of their assumptions. 
There are also methods that provide simultaneous bounds for false discovery proportions; these could in principle be used for FDX control but are not optimized for this purpose \citp{hemerik2019permutation, katsevich2020simultaneous, blanchard2020post,
 goeman2021only,blain2022notip,vesely2023permutation}.

FDX control provides an attractive statistical guarantee on the hypotheses that have been rejected. On the other hand, FDX methods tend to have relatively low power when $\gamma$ and $\alpha$ are small. For this reason, it can be useful to take $\alpha=0.5$. We then know that with probability at least $0.5$, the FDP is at most $\gamma$, i.e. $\mathbb{P}(FDP\leq \gamma)\geq 0.5$. 
This means that one  controls the median of the FDP, or \emph{mFDP} for short \citp{hemerik2024flexible}. 
mFDP control means that if the procedure is performed on many independent datasets, then the median of the resulting FDPs will be at most $\gamma$.
Note that this is similar to FDR control; the only difference is that one  controls the median instead of the mean.

\subsection{Contributions}
This paper focuses on the types of multiple testing problems discussed at the beginning of this Introduction, i.e.,  we consider one-sided or non-equivalence hypotheses, where it is expected that most hypotheses are  false.
For these settings, this paper proposes nonparametric (or semiparametric) mFDP controlling methods that are more powerful than existing exact mFDP and FDR methods. The proposed methods are built in such a way that they have good power when most parameters tend to lie in their expected directions or intervals, which correspond to the alternative hypotheses. The methods are valid regardless of the actual parameters, but they only have good power if the user's a priori expectations are mostly correct---which is usually the case in the applications we are interested in here. Our approach then tends to be superior in power, even compared to known  exact mFDP procedures based on stepping down \citp{lehmann2005generalizations} or  closed testing \citp{hemerik2018false,goeman2021only,hemerik2024flexible}.
Of course, most multiple testing procedures become more powerful if appropriate one-sided tests  instead of two-sided tests are used. However, our approach exploits such knowledge more effectively, in  settings such as those in Figure \ref{fig:scatter2}. 

The proposed methods are based on computing a test statistic for every hypothesis and  exploiting a natural symmetry property of the test statistics. For example, it is well know that test statistics are often asymptotically multivariate normal, so that they are asymptotically jointly symmetric about their means.  Even if $n$ is finite and the test statistics are not normal  then they are often still symmetric, as we will explain.
The consequence of the symmetry property is that the joint distribution of the test statistics is unchanged when we reflect them about their means. We exploit reflected test statistics to derive a \emph{median unbiased} estimator of the FDP, i.e., an estimator that underestimates the FDP with probability at most 0.5 \citp{hemerik2024flexible}.  This is the same as a $50\%$-confidence upper bound for the FDP. Note that although we write ``median unbiased'' for short, we thus mean that the estimator is ``median not-downward biased''. Having a truly median unbiased estimate of the FDP seems impossible in general, because we would need information such as how many hypotheses are false.

Later a procedure for mFDP \emph{control} is proposed,  which makes use of the mentioned FDP estimates.
This mFDP controlling procedure evaluates the FDP estimates for different rejection regions, and then chooses a rejection region in such a way that the mFDP is at most $\gamma$. 
Since the FDP estimates used are not simultaneously valid $50\%$-confidence bounds, it is nontrivial to choose the rejection region in a way that leads to valid mFDP control.
However, Theorems \ref{thmcontrol} and \ref{thmasymptcontrol} state that this can be done.  
The manner of proving these results is new to our knowledge, since it is not based on closed testing or another familiar technique. 

\subsection{Setup of this paper}
In Sections \ref{secexm} and \ref{secsingle}
we  review some existing settings where nonparametric methods are useful. Examples of tests for a single hypothesis are in  Section \ref{secsingle}. In Section \ref{secexm} we discuss multiple testing and the method SAM, which is the main nonparametric competitor for FDP estimation.

After introducing our main assumption in Section \ref{secmainas},
in Section \ref{secfdpest} we discuss median unbiased estimation of the FDP. 
For the case of one-sided testing---e.g., noninferiority testing---a novel method for median unbiased estimation of the FDP is in Section \ref{secnovel}.   Its admissibility is established in Section  \ref{secctad} and the method is conceptually compared to SAM in Section \ref{seccompSAM}.  Section \ref{secet} defines a median unbiased FDP estimator  for the case of equivalence testing. 

In Section \ref{secfdx} we discuss \emph{control} of the median of the FDP, i.e., methods which guarantee that $\mathbb{P}(FDP\leq \gamma)\geq 0.5$.
The novel mFDP controlling method is  discussed in Section \ref{seckorn} and implemented in the R package \verb|mFDP| \citp{hemerik2026mfdp}. This method builds on our median unbiased FDP estimates from Section \ref{secnovel} (directional testing) and Section \ref{secet} (equivalence testing). A link to the procedure in \citt{hemerik2024flexible} is discussed in Section \ref{secflex}. That method provides the flexibility of choosing $\gamma$ post hoc, but is less powerful for mFDP control than the novel method.
Section \ref{secsim} contains simulation and data analysis results. We end with a Discussion.

\section{Existing nonparametric  methods for FDP estimation} \label{secexm}
The methods proposed in this paper are nonparametric (or semiparametric), in the sense that the test statistics are not assumed to follow some parametric distribution. They can be used in many settings where existing nonparametric multiple testing methods can also be used. 
In Section \ref{secsingle}, we review some basic examples of nonparametric models and tests, for the case where there is a single hypothesis.
These models are extended to multiple testing settings in the examples in Section \ref{secexvalidass}.
It is worth remarking that the new multiple testing methods can also be combined with various parametric models and tests. Indeed,  test statistics are often asymptotically multivariate normal, and then the symmetry assumption that we will require (Assumption \ref{asssym} in Section \ref{secmainas}) is satisfied asymptotically.

Section \ref{secfdpest} is about median unbiased estimators of the FDP, i.e.,  estimators $\bar{FDP}$ satisfying  
\begin{equation} \label{eq:goalbound}
\mathbb{P}(FDP\leq \bar{FDP})\geq 0.5.
\end{equation}
We now discuss existing   nonparametric approaches for median unbiased estimation of the FDP. These are  compared to our novel approach conceptually in Section \ref{seccompSAM} and using simulation results in Section \ref{secsimest}.

The main nonparametric competitors of our novel mFDP estimators are the permutation methods in \citt{hemerik2018false}. 
 \citt{hemerik2018false} contains a fast method, which we will refer to as SAM, and a more powerful but much slower procedure, which we will refer to as SAM+CT. Here SAM stands for ``Significance Analysis of Microarrays'', but the methodology is very general and not at all limited to microarray data. Further, ``CT'' stands for closed testing,  which is a fundamental principle for constructing admissible multiple testing procedures \citp{genovese2006exceedance,goeman2011multiple,goeman2021only}. One message of this paper is that in the settings we are interested here, our estimators perform better than SAM and SAM+CT---in the sense that our estimates of the FDP  tend to be smaller.

Since the comparison of our estimators with SAM and SAM+CT is an important part of this paper, we summarize these methods here. 
Let $X$ be data, taking values in a sample space $\mathcal{S}$. Consider null hypotheses $H_1,...,H_m$ and corresponding test statistics $T_j: \mathcal{S} \rightarrow \mathbb{R}$, $1\leq j \leq m$, which may be dependent. 
\citt{hemerik2018false} considers quite general rejection regions which can depend on $1\leq j\leq m$, but for simplicity we will assume  the following rejection rule:  we reject all hypotheses $H_j$ with $T_j(X)>t$, where $t\geq 0$ is some prespecified value. This means that the set of  indices of the rejected hypotheses is 
$$ \mathcal{R}(t) =\mathcal{R}(X,t) =\{1\leq j \leq m: T_j(X)>t\}.$$
We will often refrain from writing the argument $X$ or $t$ when it is clear from the context.
Write $R(t)=R(X,t)=|\mathcal{R}(t)|$ and let 
$$ \N =\{1\leq j \leq m: H_j \text{ is true}\},$$ which we now assume to be nonempty for convenience. 
The number of false positive findings is then $V(t)=|\N\cap \R(t)|$ and the FDP is 
$$FDP(t) = V(t)/(R(t)\vee 1).$$

Let $\G$ be a set of transformations $g:\mathcal{S}\rightarrow\mathcal{S}$, such that $\G$ is a group under composition of maps  \citp{hemerik2018false}.
Here $\G$ may for example consist of sign-flipping maps or permutation maps as in Section \ref{secsingle}---the difference being that here the group typically acts on a whole data matrix, e.g. it may simultaneously sign-flip or permute the data. The main assumption from \citt{hemerik2018false} is Assumption \ref{assSAM} below. An example where this assumption holds is given under Assumption 1 in \citt[][p.139]{hemerik2018false}. 

\begin{assumption} \label{assSAM}
The joint distribution of the test statistics $T_j\{g(X)\}$ with $j\in \N$ and $g\in \G$  is
invariant under all transformations in $\G$ of $X$.
\end{assumption}

SAM then considers the \emph{permutation distribution of the number of rejections}. That is, let  $R^{(1)}\leq ...\leq R^{(|\G|)}$ be the sorted values $R(g(X),t)$, $g\in \G$.
Now define 
$$\bar{V}(t) =\bar{V}(X,t)  := \min\{R^{(\lceil(1-\alpha)|\G| \rceil)}, R\},$$
i.e., 
the minimum of $R(X,t)$ and the $(1-\alpha)$-quantile of the values $R(g(X),t)$, $g\in \G$. If  $\alpha=0.5$, this means taking the median of those values. Theorem 2 in \citt{hemerik2018false} states that $\mathbb{P}(V(t)\leq \bar{V}(t))\geq 1-\alpha$, which for $\alpha=0.5$ translates to 
$$\mathbb{P}(V(t)\leq \bar{V}(t))\geq 0.5.$$
Defining  $$\bar{FDP}(t) = \bar{V}(t) / (R(t)\vee 1) $$
gives a bound   $\bar{FDP}(t)$ satisfying inequality \eqref{eq:goalbound}.

As noted in \citt{hemerik2018false}, the bound $\bar{FDP}(t)$ can be uniformly improved, while still guaranteeing that   inequality \eqref{eq:goalbound} holds. This improvement is based on \emph{closed testing}. The improvement is provided in Proposition 1 of  \citt[][]{hemerik2018false} and in that paper the details are given.
This is the method that we will refer to as SAM+CT.  Intuitively, the idea is that especially when there are many false hypotheses, the basic SAM method will be conservative because the many false hypotheses cannot lead to false discoveries. SAM+CT addresses that issue. 

As discussed Appendix \ref{seccompSAM} and in \ref{secsimest},
our novel FDP estimator does not uniformly improve SAM+CT, but on average it povides lower FDP estimates, for the scenarios that we are interested in here. Another crucial advantage of our novel method is that it is much faster than SAM+CT. Whereas the computation time of SAM+CT is typically exponential $m$, the computation time of the novel method in linear in $m$, after sorting the p-values. In practice SAM+CT is only computationally feasible for quite small multiple testing problems ($m\leq100$). 

\section{Hypotheses and symmetry assumption} \label{secmainas}

In this paper we study multiple hypothesis testing problems of the following form.
Consider hypotheses $H_1,...,H_m$ and  corresponding  test statistics  $T_1,...,T_m$ with means $\mu_1,...,\mu_m\in \mathbb{R}$ respectively. These test statistics might e.g. be unbiased estimators of certain parameters $\theta_1,...,\theta_m$ of interest, in which case we simply have $\mu_j= \mathbb{E}(T_j)=\theta_j$, $1\leq j \leq m$.
Throughout Section \ref{secnovel}  we suppose that for every $1\leq j\leq m$ the null hypothesis is
 \begin{equation} \label{defoneside}
H_j:  \mu_j\leq \delta_j,
\end{equation}
where  $\delta_1,...,\delta_m\in \mathbb{R}$ are constants.  Testing non-equivalence hypotheses of the form  $$H_j:  |\mu_j|\geq \delta_j,$$
with $\delta_j>0$ for all $j$,
 will be covered in Section \ref{secet}.

Throughout the rest of this paper, we will assume   that the test statistics $T_1,...,T_m$ satisfy the following symmetry property.
In fact, it is only an assumption on the joint distribution of $(T_j: i\in \N)$, i.e., the   test statistics corresponding to the true hypotheses.

\begin{assumption} \label{asssym}
The vector $(T_j-\mu_j : j\in \N)$  is nondegenerate and satisfies  $$(T_j-\mu_j : j\in \N) \,{\buildrel d \over =}\, (-(T_j-\mu_j): j\in \N).$$
\end{assumption}
The above assumption  is reasonable in many situations. For example, if the statistics  $(T_j : j\in \N)$ are multivariate normal, then the assumption is satisfied. Note that even if that is not the case, then the  test statistics are often multivariate normal asymptotically. If we only require asymptotic validity of the proposed methods, we can replace Assumption \ref{asssym} by the assumption that it holds asymptotically.
 Further, we do not need normality at all. 
 Section \ref{secexvalidass} contains two basic examples where the statistics are not normal, but Assumption \ref{asssym} is satisfied for finite samples.

\section{FDP estimation} \label{secfdpest}
In this section we cover the proposed median unbiased FDP estimators.  Besides being useful in their own right, the estimators will be used by the mFDP controlling method from Section \ref{seckorn}.

\subsection{FDP estimator for one-sided tests} \label{secnovel}
In this section we define our median unbiased FDP estimator for the case of one-sided tests. 
In Section \ref{secctad} it will be proved that the estimator is admissible, i.e., it cannot be uniformly improved by a different median unbiased estimator.

For every $1\leq j \leq m$ let $\delta_j\in \mathbb{R}$, and pick $t\geq 0$. Consider hypotheses of the form \eqref{defoneside}.
Reject the hypotheses with indices in 
$$\mathcal{R}(t)=\mathcal{R}((T_1,...,T_m),t) = \{1\leq j \leq m: T_j-\delta_j > t\}.$$
Instead of  $\mathcal{R}((T_1,...,T_m),t)$, we will often write $\mathcal{R}(t)$ or $\mathcal{R}$ when it is clear from the context that the arguments are $(T_1,...,T_m)$ and $t$.
The number of false discoveries is then  $V = V(t) = \N\cap \mathcal{R}(t)$, where we used the notation introduced in Section \ref{secexm}. Also recall the notation $R=|\mathcal{R}|$. Write $a_1\wedge a_2$ for the minimum of two numbers $a_1$ and $a_2$.
Our estimator for the number of false positives is defined as 
$$
\tilde{V}(t) =  R^-(t)  \wedge  R(t),
$$
where $R^-(t) = | \R^-(t) |$ with
$$ \R^-(t)  := \{1\leq j \leq m: T_j-\delta_j<-t\}.$$
Analogously, our FDP estimator is defined to be
$$   \tilde{FDP}(t) =  \tilde{V}(t)/(R(t)\vee 1) .$$
The way in which we thus estimate the FDP may bring to mind the knockoffs procedure  \citt{barber2015controlling,candes2018panning}. Thus, the estimator  $\tilde{FDP}(t)$ is in itself not really novel. However, a difference with knockoffs is that the null test statistics based on knockoffs are mutually uncorrelated, whereas here they need not be.

We first establish that 
 under Assumption \ref{asssym}, these estimators are median unbiased, i.e., they satisfy
  \begin{equation} \label{eq:mainprV}
 \mathbb{P}\big\{V(t)\leq \tilde{V}(t)\big\}\geq 0.5,
 \end{equation}
 or equivalently,
 \begin{equation} \label{eq:mainpr}
 \mathbb{P}\big\{FDP(t)\leq \tilde{FDP}(t)\big\}\geq 0.5.
 \end{equation}
\begin{theorem} \label{thmnewvalid}
Under Assumption \ref{asssym}, the bound $\tilde{FDP}$ satisfies equation \eqref{eq:mainpr}.
\end{theorem}
 Proofs are in Appendix \ref{appproofs}.
The estimator $\tilde{V}$ is conceptualy compared with SAM in Appendix \ref{seccompSAM}. There, an interesting property is discussed, namely that if $\alpha=0.5$, using a very small number of permutations in SAM may improve its power.
  
We now investigate whether the estimator $\tilde{V}(t)$ is admissible, i.e., whether it can be uniformly improved, without violating requirement \eqref{eq:mainprV}. 
Assume for a moment that $\mu_j=\delta_j$ for all $1\leq j \leq m$, so that all hypotheses are true, but barely so.
Then we can note the following.

\begin{lemma} \label{lemmaties}
Suppose Assumption \ref{asssym} holds and  $\mu_j=\delta_j$ for all $1\leq j \leq m$. Hence $\N=\{1,...,m\}$ and $V(t)=R(t)$.
Define the events
\begin{align*}
& E_1 = \{V(t) <  R^-(t)\} =  \{R(t) <  R^-(t)\}, \\
& E_2 = \{V(t) =  R^-(t)\} =  \{R(t) =  R^-(t)\}, \\
&E_3 = \{V(t) >  R^-(t)\} =   \{R(t) >  R^-(t)\}.
\end{align*}
Then $\mathbb{P}(E_1) = \mathbb{P}(E_3)$
and $\mathbb{P}(E_1) =   0.5  - 0.5\mathbb{P}(E_2)$.
\end{lemma}

Thus, in the scenario of Lemma \ref{lemmaties}, we will only have  $\mathbb{P}(E_1)=0.5$ if  $\mathbb{P}(E_2)=0$, which is however typically not the case---although it is often approximately true when there  are many hypotheses.
In conclusion, $\mathbb{P}(E_3)$ is typically strictly smaller than $0.5$ and consequently $\pr(V(t)> \tilde{V}(t)) $ is typically strictly smaller than $0.5$. It can be seen that this is also the case for other values of $\delta_1,...,\delta_m$.

This suggests to slightly adapt the procedure in the situation that  $R(t) $ and $ R^-(t)$ are tied, in order to have a method that can be shown to be admissible. Indeed, consider the following adapted bound $\tilde{V}'$:
$$
    \tilde{V}'(t)= 
\begin{cases}
   d  & \text{if } R(t) =  R^-(t) \\
    \tilde{V}(t)          & \text{otherwise,}
\end{cases}
$$
where $d$ is an independent  variable that equals $\tilde{V}(t)$ or $-\infty$ both with probability $0.5$.
This procedure sets the bound to $-\infty$ with probability 0.5 when $R(t) =  R^-(t)$. 
Since the bound $\tilde{V}'(t)$ is based on a random coin flip $d$, we
do not recommend using $\tilde{V}'(t)$ in practice. However, we will show that $\tilde{V}'(t)$ is admissible. 
Since typically $\tilde{V}(t)=\tilde{V}'(t)$ with high probability when there are many rejected hypotheses, it follows that $\tilde{V}$ is essentially admissible for most practical purposes.

The following result states that the bound $\tilde{V}'(t)$ is valid. 

\begin{proposition} \label{basicsvbarprime}
Suppose  Assumption \ref{asssym} holds.
The bound $\tilde{V}'(t)$ is median unbiased in the sense that $\mathbb{P}\{V(t) \leq   \tilde{V}'(t)\}\geq 0.5.$
Further, in case  $\mu_j=\delta_j$ for all $1\leq j \leq m$, we   have  $\mathbb{P}\{V(t) \leq   \tilde{V}'(t)\} = 0.5 $.
\end{proposition}

Proposition \ref{basicsvbarprime} implies that in some cases $\mathbb{P}\{V(t) \leq   \tilde{V}'(t)\} = 0.5 =\alpha$, i.e., in some cases $\tilde{V}'(t)$ exhausts $\alpha$. However, this does not yet imply that $\tilde{V}'(t)$ is admissible. 
Indeed, multiple testing methods that exhaust $\alpha$ when $\N=\{1,...,m\}$, are not always admissible.
For example, Bonferroni also exhausts $\alpha$ in certain situations where $\N=\{1,...,m\}$, yet it is not admissible since it is uniformly improved by Bonferroni-Holm \citp{holm1979simple, goeman2014multiple}. 
Nevertheless, it turns out that the simple estimator $\tilde{V}'(t)$ is admissible.
In Section \ref{secctad}
 we link the bounds $\tilde{V}$ and $\tilde{V}'$ to the theory of \emph{closed testing} and prove that the bound $\tilde{V}'$ is admissible, i.e., that it cannot be uniformly improved. Thus, by Section \ref{secctad}, we have the following admissibility result.

\begin{theorem} \label{thmadm}
Assume the distribution of $(T_1,...,T_m)$ satisfies Assumption \ref{asssym}.
Then the bound $\tilde{V}'(t)$ is admissible, i.e., there exists no method producing bounds $\tilde{V}^*(t)$ based $T_1,...,T_m$ with the following properties: 
\begin{enumerate}
\item The bound $\tilde{V}^*(t)$ is valid, i.e.,   it generally holds that  $\mathbb{P}\{V(t)\leq \tilde{V}^*(t)\}\geq 0.5$ under Assumption \ref{asssym}.
\item $\tilde{V}'(t) \geq  \tilde{V}^*(t)$ almost surely, i.e., $\tilde{V}^*(t)$ is never higher than $\tilde{V}'$ under Assumption \ref{asssym}.
\item For at least one distribution satisfying Assumption \ref{asssym}, $\mathbb{P}\{ \tilde{V}'(t) >\tilde{V}^*(t)\}>0$, i.e., $\tilde{V}^*$ improves $\tilde{V}'$.
\end{enumerate}
\end{theorem}

\subsection{FDP estimator for equivalence testing} \label{secet}
%

In Section \ref{secnovel}  we have defined and studied the median unbiased FDP estimator $\tilde{V}$ for the case of directional testing.
Here, we provide an estimator for the case of equivalence testing. 
It will be used by the mFDP controlling method from Section \ref{seckorn}.
Here, we consider multiple null hypotheses of non-equivalence 
$$H_j: |\mu_j|\geq \delta_j, \quad 1\leq j \leq m,$$
with $\delta_j>0$ for every $1\leq j \leq m$.
We say that $\mu_j$ lies in the \emph{equivalence region} if $\mu_j\in (-\delta_j,\delta_j)$.

For $t\geq 0$ define
\begin{align*}
\R(t) :=& \{1\leq j \leq m: |T_j|<\delta_j -t  \}, \qquad R(t):= |\R(t)|, \\
\R^-(t):=&   \{1\leq j \leq m: |T_j|>\delta_j + t  \}, \qquad R^-(t):= |\R^-(t)|.  
\end{align*}
We reject all hypotheses with indices in $\R(t)$. Thus, we only reject the hypotheses $H_j$ for which $T_j$ lies far enough within the equivalence region $(-\delta_j,\delta_j)$. Clearly, for $t\geq \max\{\delta_j: 1\leq j \leq m   \}$,  $R(t)$ is simply 0.
The number of false discoveries is as usual $V(t)=\N\cap \R(t) $.
Further, let $\tilde{V} :=   R^- \wedge R$.  Thus,  in the present section, we redefine several of the symbols from Section \ref{secnovel}. The following theorem states that in the present  setting, $\tilde{V}$ a median unbiased  estimator of $V$ under  Assumption  \ref{asssym} from Section \ref{secmainas}.

 \begin{theorem} \label{boundetvalid}
Assume  $(T_i:i\in \N)$ satisfies Assumption \ref{asssym}. 
 Then $\mathbb{P}(V(t)\leq   \tilde{V}(t)) \geq 0.5.$
 \end{theorem}
 
It turns out that like in Section \ref{secnovel} (one-sided testing),  in the present setting (equivalence testing) $\tilde{V}$  coincides with a closed testing procedure.  This is further discussed in Section \ref{appVtildectp}.

\section{mFDP control} \label{secfdx}
We now focus on achieving control of the median of the FDP, or mFDP control for short. This means guaranteeing that $\mathbb{P}(FDP\leq \gamma ) \geq 0.5$.  We discuss two methods. The first one (see Section \ref{seckorn}) is novel and usually most powerful for mFDP control (as confirmed in Section \ref{secsimcon}). The second method (Section \ref{secflex}) is partly novel, since it builds on theory from \citt{hemerik2024flexible}. This procedure has the added flexibility that $\gamma$ can be chosen after seeing the data. The price to pay for this added flexibility is lower power.

\subsection{Method 1: novel mFDP  controlling procedure} \label{seckorn}  
The method that we propose here relies on the $50\%$-confidence bounds $\tilde{V}$ that have been proposed in Sections \ref{secnovel} and \ref{secet} for one-sided testing and equivalence testing respectively. The rough idea is to choose the rejection threshold in such a way that the FDP is estimated to be at most $\gamma$ for this threshold and all stricter thresholds. This is similar to the idea underlying the methods in \citt{korn2004controlling}, \citt{romano2007control} and \citt{hemerik2025resampling}. A major difference with those procedures  is that they are  often not as powerful when there are many false hypotheses.

As in Sections \ref{secnovel} and \ref{secet}, we consider the bound $\tilde{V}(t)$, which is a function of the threshold $t\in [0,\infty)$.
In general, $\R$ and $\R^-$  are functions $[0,\infty ) \rightarrow \{I: I\subseteq \{1,...,m\}\}$ and $R$, $V$ and $\tilde{V}$ are functions $[0,\infty ) \rightarrow \mathbb{N}$. Further, 
$$\tilde{FDP}(t):=\tilde{V}(t)/(R(t)\vee 1)$$
is a function $[0,\infty) \rightarrow [0,1]$.
Note that the meanings of these symbols depend on whether we are in the setting of Sections \ref{secnovel} or Section \ref{secet}.
In general,  the functions $R(\cdot)$, $V(\cdot)$ and $\tilde{V}(\cdot)$ are data-dependent  step functions with a finite number of discontinuities. 
Assume for convenience that with probability 1, $|T_1\pm\delta_1|,...,|T_m\pm\delta_m|$ are all distinct.
It follows  that with probability 1, the discontinuities of $R(\cdot)$ and $\tilde{V}(\cdot)$ are disjoint.

Both in Section \ref{secnovel} and Section \ref{secet},  $R(t)$,  $R^-(t)$, $V(t)$ and $\tilde{V}(t)$ are non-increasing in $t$. This is not generally true  however for the bound $\tilde{FDP}(t)$. Indeed, if  $\tilde{V}(\cdot )$ has a discontinuity at
 $t\in \co$, then $\tilde{FDP}(\cdot)$ jumps down at $t$, but if  $R(\cdot )\vee 1$  has a  discontinuity at $t$, then $\tilde{FDP}(\cdot)$ jumps up at $t$. However, if many hypotheses are false, then the tendency is for $\tilde{FDP}(\cdot)$ to roughly decrease in $t$, as illustrated in Figure \ref{fig:constructionofs}. Indeed, the larger $t$ is, the stricter the method is and the smaller $FDP(t)$ tends to be. 
  \black A further observation is that $R(\cdot )$, $\tilde{V}(\cdot )$ and hence $\tilde{FDP}(\cdot )$  are right-continuous.

\subsubsection{Definition of the procedure}
Let $\M' \subseteq (0,\infty)$ be the set of arguments $t$ where  ${R}(\cdot )$ or $R^-(\cdot )$ has discontinuities. Note that $\pr(\M'\neq\emptyset)=1$.
Let $\M= \{0\}\cup\M'$.
Define 
$$
s =\max\{t\in\M: \tilde{FDP}(t) >\gamma   \}, 
$$
where the maximum of the empty set is $-\infty$.
Note that $s$ is strictly smaller than $M:=\max(\M)$, since $\tilde{FDP}(M) =R(M)=0$. Let
\begin{equation} \label{splusasmin}
s^+=\min\{t\in \M: t>s\}.
\end{equation}
Note that $s < s^+$ and often $s \approx s^+$. Further, $\tilde{FDP}(s^+)\leq \gamma.$  Since $\tilde{FDP}(\cdot )$ is a step function that is continuous from the right, $\tilde{FDP}(t)=\tilde{FDP}(s)>\gamma$ for all $t\in [s,s^+)$. It follows that 
\begin{equation} \label{splusassup}
s^+ =\sup\{t\geq 0: \tilde{FDP}(t)>\gamma\}.
\end{equation}
The construction of $s^+$ is illustrated in Figure \ref{fig:constructionofs}.
While the formulation \eqref{splusassup} is elegant, the formulation \eqref{splusasmin} is useful from a computational perspective.

The \emph{proposed procedure rejects all hypotheses with test statistics strictly larger than $s^+$,} i.e., it rejects all hypotheses with indices in $\R(s^+)$. This procedure is implemented in the functions \emph{mFDP.direc()} (for directional testing) and \emph{mFDP.equiv()} (for equivalence testing) of the R package \verb|mFDP| \citp{hemerik2026mfdp}.

\begin{figure}[ht] 
\centering
  \includegraphics[width=0.6\linewidth]{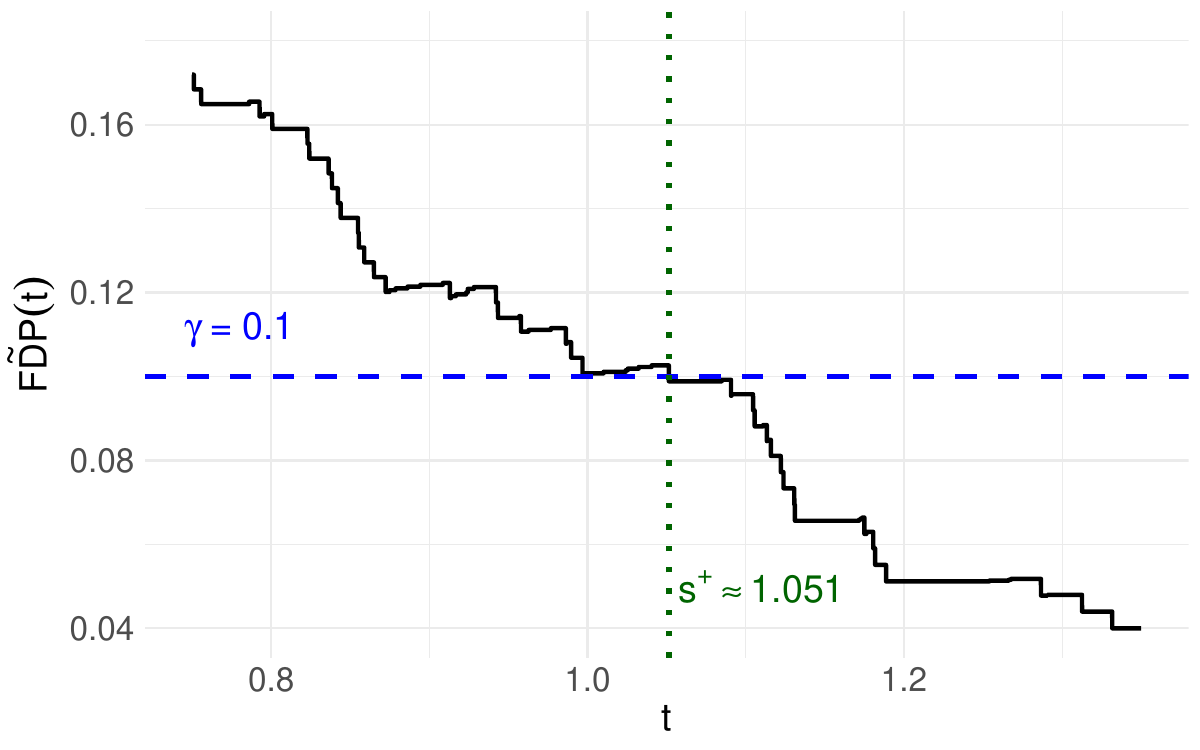}
  \caption{\emph{This figure illustrates the construction  of $s^+$. The smaller the target FDP $\gamma$ is chosen to be, the larger---i.e., stricter---the threshold $s^+$ is. In case of directional testing, our method rejects all hypotheses $H_j$ with $T_j>\delta_j+s^+$. In case of equivalence testing, it rejects all hypotheses $H_j$ with $|T_j|<\delta_j-s^+$.}
} \label{fig:constructionofs}
\end{figure}

\subsubsection{Properties of the procedure}
Obviously, our procedure provides valid mFDP control if and only if 
\begin{equation} \label{eq:askorn}
\pr\{  FDP(s^+) \leq \gamma   \} \geq 0.5.
\end{equation}
A somewhat stronger property is 
\begin{equation} \label{eq:askorn2}
\pr\{\forall t\geq s^+: FDP(t) \leq \gamma   \} \geq 0.5,
\end{equation}
which means that we have mFDP control simultaneously over all thresholds $t\geq s^+$.

In simulations, we did not find any setting where condition \eqref{eq:askorn} or  \eqref{eq:askorn2} was violated.
In addition, we now provide two theorems that state properties under which condition \eqref{eq:askorn2} is provably satisfied, for finite samples or asymptotically.

\begin{theorem} \label{thmcontrol}
Suppose Assumption \ref{asssym} holds.
Assume that all hypotheses are true or that $(T_i:i\in \N)$ is independent of $(T_i:i\in \N^c)$, where $\N^c=\{1,...,m\}\setminus \N$.
Then property \eqref{eq:askorn2} and hence \eqref{eq:askorn} is satisfied, both in the setting of Section \ref{secnovel}, i.e. directional testing, and in the setting of  Section \ref{secet}, i.e. equivalence testing.
In the former setting, this means that  if we reject all hypotheses $H_j$ with $T_j>\delta_j+s^+$, then property \eqref{eq:askorn2} holds. In the latter setting,   this means that property \eqref{eq:askorn2} holds  if we reject all hypotheses $H_j$ with $|T_j|<\delta_j-s^+$.
\end{theorem}

The additional assumption in the above theorem is in theory unpleasant, but we did not find a simulation scenario where it was necessary. Moreover, existing FDX methods either make dependence assumptions as well, or  have  low power \citp{delattre2015new}. The following theorem guarantees asymptotic mFDP control without requiring this assumption.

\begin{theorem} \label{thmasymptcontrol}
Suppose that the test statistics $(T_1,...,T_m)=(T_1^n,...,T_m^n)$ depend on an index $n$, which typically denotes the sample size.
Suppose we are in the  setting of Section \ref{secnovel}, i.e. directional testing. Assume that for every $j\in \N^c$, as $n\rightarrow\infty$,   $T_j$ converges to $\infty$ in probability.  Further, suppose that for every $j\in \N$, $T_j$ stays bounded from above in probability as $n\rightarrow\infty$.
The property \eqref{eq:askorn2} and hence \eqref{eq:askorn} are asymptotically satisfied, i.e., 
$$\liminf_{n\rightarrow\infty} \pr\big\{\forall  t\geq s^+: FDP(t)\leq \gamma  \big\} \geq 0.5.$$
\end{theorem}

\begin{example}
Consider the situation that we are interested in hypotheses $H_1,..., H_m$ of the form $H_j: \theta_j \leq c_j$. Suppose we have corresponding right-sided t-statistics $T_1,...,T_m$. Then for each $j\in \N$, $\mu_j=\mathbb{E}(T_j)\leq 0$, while we expect the test statistics $T_j$ with $j\in \N^c$ to become very large as the sample size increases. The hypotheses are of the required form \eqref{defoneside} if we take $\delta_1=...=\delta_m=0$. Thus, Theorem \ref{thmasymptcontrol} applies.
\end{example}

A notable property of our methods is that if $\R(0)=\{1,...,m\}$---i.e., all test statistics lie in the noninferiority  or equivalence region (e.g. the shaded regions in Figure \ref{fig:scatter2})---then   $\tilde{V}(t)=0$ for all $t\geq 0$ and  $\R(s^+)=\{1,...,m\}$, i.e., we reject all hypotheses, regardless of $\gamma$.
Indeed, if $\R(0)=\{1,...,m\}$, then clearly $\tilde{V}(t)=0$ for all $t\geq 0$ and hence $s=-\infty$, so that $s^+=0$ and $\R(s^+)=\{1,...,m\}$. This property is an extreme example of a more general property of our methods: if $\R(0)$ is large, then the bounds $\tilde{V}(t)$ are small and our mFDP controlling procedure tends to reject many hypotheses. This explains the good simulation results in  Sections \ref{secsimest} and \ref{secsimcon}.

\subsection{Method 2: flexible mFDP  control} \label{secflex}
The method in \citt{hemerik2024flexible} provides \emph{flexible control of the mFDP}. This means that the procedure not only ensures that $\pr\{FDP>\gamma\}\leq 0.5$ but in addition allows the user to freely choose $\gamma\in [0,1)$ based on the data, without invalidating mFDP control. This added flexibility comes with a price in terms of power. Indeed, this procedure is usually less powerful than the method from Section \ref{seckorn}, as confirmed with simulations in Section \ref{secsimcon}. The flexibility however clearly also has advantages, as discussed in \citt{hemerik2024flexible}. Note that the new method also provides some flexibility, in the sense that  property \eqref{eq:askorn2} says that we have mFDP control simultaneously over all $t\geq s^+$.

A seemingly major difference between the method of \citt{hemerik2024flexible} and those that have been proposed in this paper, is that the former is based on $\pvs$, rather than directly using a symmetry property of the test statistics. However, as shown there, it is actually not necessary to have valid $\pvs$, i.e., they need not be uniformly distributed  under the null. Moreover, we will show that under Assumption \ref{asssym}, we can transform our test statistics into $\pvs$ that can validly be used within the procedure of \citt{hemerik2024flexible}. 
In particular,  even if we do not known the null distributions of our test statistics $T_1,...,T_m$, then we can simply guess their symmetric null distributions and validly use the resulting invalid $\pvs$ in  the method of \citt{hemerik2024flexible}. The reason is that is suffices that the null $\pvs$ stochastically dominate a distribution with a certain symmetry property.  Indeed, letting $N=|\N|$, it suffices that the distribution of $(P_j: j \in \N)$ is stochastically larger than a random vector $(Q_1,...,Q_{N})$ taking values in $\mathbb{R}^N$ that satisfies
\begin{equation} \label{sympvs}
(Q_1,...,Q_N) \,{\buildrel d \over =}\,   (1-Q_1,...,1-Q_N),
\end{equation}
i.e.,
$$  (Q_1-0.5,...,Q_N-0.5) \,{\buildrel d \over =}\,   (0.5-Q_1,...,0.5-Q_N).$$

The following theorems state that if we construct potentially invalid $\pvs$  in the way just mentioned, then the methods in \citt{hemerik2024flexible}  can be validly applied based on these   $\pvs$. Theorem \ref{pvsgood} is about directional testing and Theorem \ref {pvsgoodet} is about equivalence testing.

\begin{theorem} \label{pvsgood}
Consider the setting of Section \ref{secnovel} (directional testing) and suppose Assumption \ref{asssym} holds.
For every $1\leq j \leq m$, consider some probability density function $f_j: \mathbb{R}\rightarrow \mathbb{R}^+ $ that is symmetric about 0.
Here, $f_j$ is the guessed  or perhaps exactly known null distribution of $T_j-\mu_j$. 
For every $1\leq j \leq m$, consider the $\pv$ 
$P_j = \int_{T_j-\delta_j}^{\infty} f_j(x) dx$.

Then, $(P_j:j \in \N)$ is deterministically larger than or equal to some vector satisfying property \eqref{sympvs}. As a consequence, the $\pvs$ can validly be used in the methods of  \citt{hemerik2024flexible}. In particular, as in \citt{hemerik2024flexible}, let $\cop\subset[0,1]$ be the range of $\pv$ thresholds of interest, where usually a good choice is $\cop=[0,0.5]$ or some smaller set.
Then 
 the bound $\tilde{B}: \cop \rightarrow \mathbb{N} $ constructed in  Section 3.3 of that paper satisfies
$$\pr\Big[ \bigcap_{v\in\cop} \big|\{j\in \N: P_i< v\}\big|\leq \tilde{B}(v) \Big]\geq 0.5,$$
i.e., the values $\tilde{B}(v)$ are simultaneous $50\%$ confidence bounds for the number of false positives $|\{j\in \N: P_i< v\}|$.
In particular, we can choose any $\gamma\in [0,1)$ after seeing the data and choose any rejection threshold $t\in \cop$ such that  $\tilde{B}(t) /  (|\{1\leq j \leq m: P_i< t\}\big|\vee 1)\leq \gamma$. This will guarantee that the median of the FDP is at most $\gamma$, i.e. $\pr(FDP(t)\leq \gamma)\geq 0.5$.
\end{theorem}

Note that in the theorem, the guessed null distributions $f_j$ are not required to be correct, although it can be advantageous power-wise if they are roughly correct.
The $\pvs$ from Theorem \ref{pvsgood} can directly be provided as arguments within the function \emph{mFDP.adjust()} of the R package \verb|mFDP|  \citp{hemerik2026mfdp}. This function computes \emph{adjusted} $\pvs$; rejecting the hypotheses with adjusted $\pvs$ smaller than $\gamma$ guarantees that $\pr(FDP(t)\leq \gamma)\geq 0.5$.

We now formulate a similar result for the case of equivalence testing.

\begin{theorem} \label{pvsgoodet}

Consider the setting of Section \ref{secet} (equivalence testing) and suppose Assumption \ref{asssym} holds.
For every $1\leq j \leq m$, consider some probability distribution function $f_j: \mathbb{R}\rightarrow \mathbb{R}^+$ that is symmetric about 0.
Here, $f_j$ is the guessed  or perhaps exactly known null distribution of $T_j-\mu_j$. 
For every $1\leq j \leq m$, define
\begin{equation} \label{TOSTpv}
P_j =
\begin{cases}
\int_{T_j+\delta_j }^{\infty } f_j(x) dx     & \text{if } T_j<0 \\
   \int_{-\infty}^{T_j-\delta_j} f_j(x) dx     & \text{if } T_j\geq0.
\end{cases}
\end{equation}
Then, $(P_j:j \in \N)$ is  larger than or equal to some vector satisfying property \eqref{sympvs}. As a consequence the results from Theorem \ref{pvsgood} apply here as well.
\end{theorem}

\begin{remark} \label{remTOST}
In Theorem \ref{pvsgoodet}, $P_j$ is the so-called \emph{TOST-p-value}, where ``TOST'' stands for ``two one-sided tests'' \citp{meyners2012equivalence,lakens2017equivalence}. This is the most common \emph{p}-value used for equivalence testing.  In appendix \ref{appproofs}, after the proof of Theorem \ref{pvsgoodet}, we prove that $P_j$ is indeed the TOST-\emph{p}-value.
\end{remark}

\section{Simulations and data analysis} \label{secsim}
In  Sections \ref{secnovel} and \ref{secet} we discussed median unbiased estimation of the FDP. The method defined in Section \ref{secnovel} is compared to alternatives in Section \ref{secsimest} below. Our mFDP controlling method from Section \ref{seckorn} is compared with competitors in Section \ref{secsimcon}. Thus, Section \ref{secsimest} is about \emph{estimation} of the FDP
 and Section \ref{secsimcon} is about \emph{controlling} the mFDP, i.e., keeping the mFDP below a prespecified value $\gamma$. Finally, in Section \ref{secdata} we analyze real data.

\subsection{Median unbiased estimation of the FDP} \label{secsimest}
The main competitor of our method from Section \ref{secnovel} is the SAM procedure from  \citt{hemerik2018false}, which uses permutations or other transformations to provide a median unbiased estimate of the FDP. Here we will show that for the settings that we are interested in in this paper, i.e. settings with parameters trending in a particular direction, our approach outperforms SAM. For SAM, the transformations that we use below are sign-flips. In case  of permutations (of cases and controls), similar results were obtained, i.e., our method clearly outperformed SAM. 

To generate test statistics, we first simulated an $n\times m$ matrix of independent standard normal variables, with $n=10$ and $m=500$. In case the data were not normal but had some other symmetric distribution, we obtained comparable results, so here we will stick to normal data. We allowed the test statistics to be correlated, which we achieved by sampling $n$ values and adding the $i$-th value to all elements of the $i$-th row of the data  matrix, $1\leq i \leq n$. This leads to a homogeneous correlation structure, where we denote the correlation by $\rho$. We also simulated (negatively correlated) blocks of positively correlated test statistics, which gave similar results to those shown here. To the columns corresponding to the false hypotheses we added a constant $d>0$. The null hypotheses were $H_j: \mu_j\leq \delta$, $1\leq j \leq m$, where $\delta:=0$.

In Figure \ref{fig:meanest} the average estimate $\tilde{FDP}$ of our method is compared to SAM's median unbiased estimate. Here $\pi_0:=|\N|/m$ is the number of true hypotheses divided by the total number of hypotheses. It can be seen that our method provided on average lower estimates of the FDP when there were many false hypotheses ($\pi_0=0.1$), which is the situation we are interested in in this paper.

\begin{figure}[ht!] 
\centering
  \includegraphics[width=\linewidth]{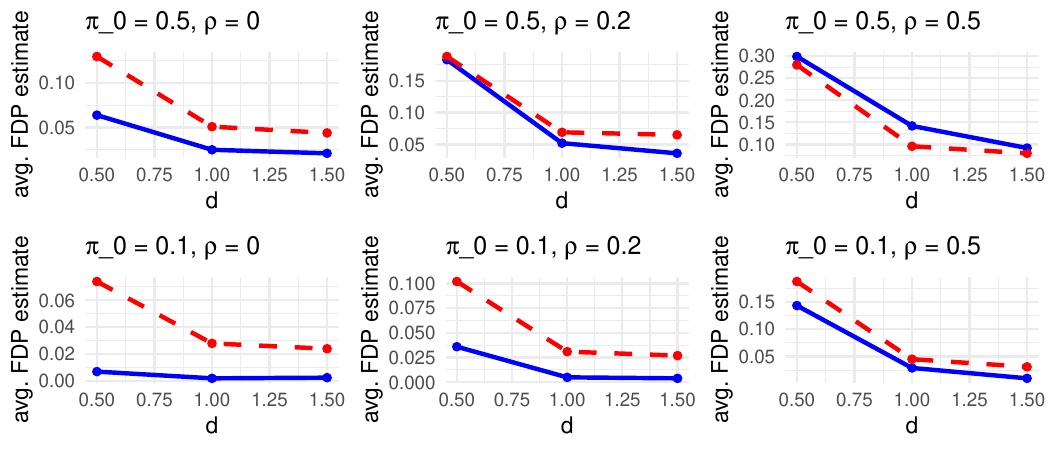}
  \caption{ \emph{The average FDP estimate of our novel estimator $\tilde{FDP}$ from Section \ref{secnovel} (solid lines) and SAM (dashed lines) as depending on  the fraction $\pi_0$ of true hypotheses, the homogeneous correlation $\rho$ between the test statistics and the effect size $d$. The number of hypotheses was $m=500$. Each estimate is based on $5\cdot 10^3$ simulations.}} \label{fig:meanest}
\end{figure}

As discussed in Section \ref{secexm}, SAM is uniformly improved by SAM+CT. However, SAM+CT is only computationally feasible when $m$ is not too large. Hence we took $m=50$ in the simulations of SAM+CT. A comparison of the novel method, SAM and SAM+CT is in Figure \ref{fig:meanestCT}.  Note that SAM+CT clearly performed better than SAM.
Further, we see that our method performed relatively well when $\pi_0$ was small, which is the setting we are interested in here. The additional advantage of our method is that it is very fast, unlike SAM+CT, which is computationally infeasible for large problems.

\begin{figure}[ht!] 
\centering
  \includegraphics[width=\linewidth]{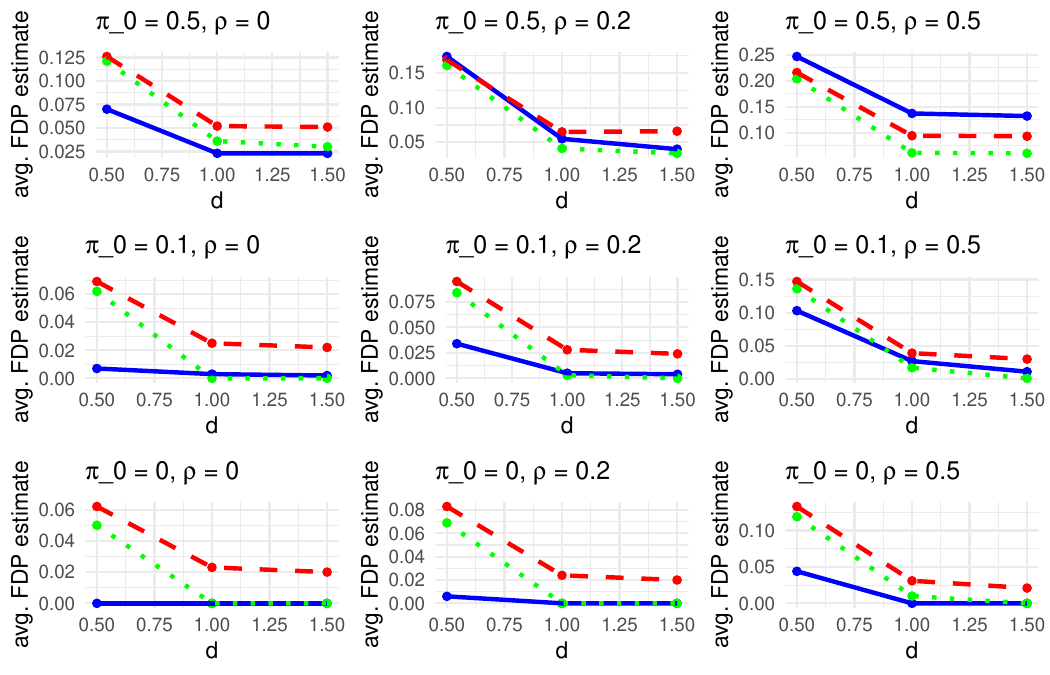}
  \caption{ \emph{The average FDP estimate of our novel estimator $\tilde{FDP}$ from Section \ref{secnovel} (solid lines), SAM (dashed) and SAM+CT (dotted) as depending on   the fraction $\pi_0$ of true hypotheses, the correlation $\rho$ between the test statistics and the effect size $d$. The number of hypotheses was taken to be $m=50$, since SAM+CT is computationally infeasible for large $m$. Each estimate is based on $10^3$ simulations.}} \label{fig:meanestCT}
\end{figure}

\subsection{mFDP control} \label{secsimcon}

Here, we show that the novel mFDP controlling procedure from Section \ref{seckorn} often has better power than some major competitors. The first method that we compare with, is the procedure from \citt{hemerik2024flexible} discussed in Section \ref{secflex}. This method is implemented in the function \emph{mFDP.adjust} from the R package \verb|mFDP| \citp{hemerik2026mfdp}. 
It is relevant to mention that that method is already compared with the procedures from \citt{goeman2019simultaneous} and \citt{katsevich2020simultaneous} in simulations in \citt{hemerik2024flexible}.

In addition, we compare with the Romano-Wolf stepdown method \citp[][Algorithm 4.1]{romano2007control}. This is and FDX method and $\alpha=0.5$ was used to obtain mFDP control. 
The Romano-Wolf method iteratively uses k-FWER methods. The k-FWER used within Romano-Wolf was the usual stepdown k-FWER method based on the $\pvs$, see e.g. \citt[][p.1142-1143]{lehmann2005generalizations}. The $\pvs$ used  were right-sided and were computed by comparing the test statistics with the standard normal distribution. Note that the Romano-Wolf method is in this case equivalent to the most powerful method among the  procedures in \citt{lehmann2005generalizations}, namely the one from pp.1147-1148. This procedure requires some assumptions on the dependence structure and is more powerful than the methods in \citt{romano2006stepdown} and \citt{romano2006stepup}.

Further,  we compare with the Benjamini-Hochberg procedure, which ensures that the mean rather than the median of the FDP is below $\gamma$ \citp{benjamini1995controlling}. The  Benjamini-Hochberg method is also based on $\pvs$. The $\pvs$ were computed in the same way as for Romano-Wolf.

The data were simulated as in Section \ref{secsimest}. For each method, the target $FDP$, $\gamma$---often called $\alpha$ in the context of Benjamini-Hochberg---was set to 0.1. We computed the power as the mean fraction of false hypotheses that were rejected. We considered effect sizes $d\in\{1,2,3\}$. Again, we focused on settings where $\pi_0$ was far from 1, which are the settings we are interested in in this paper.
The results are in Figure \ref{fig:pow}.
We see that our method had better power than the other methods, although naturally the difference with  the method from \citt{hemerik2024flexible} was small when that method had power near $100\%$.

\begin{figure}[ht!] 
\centering
  \includegraphics[width=\linewidth]{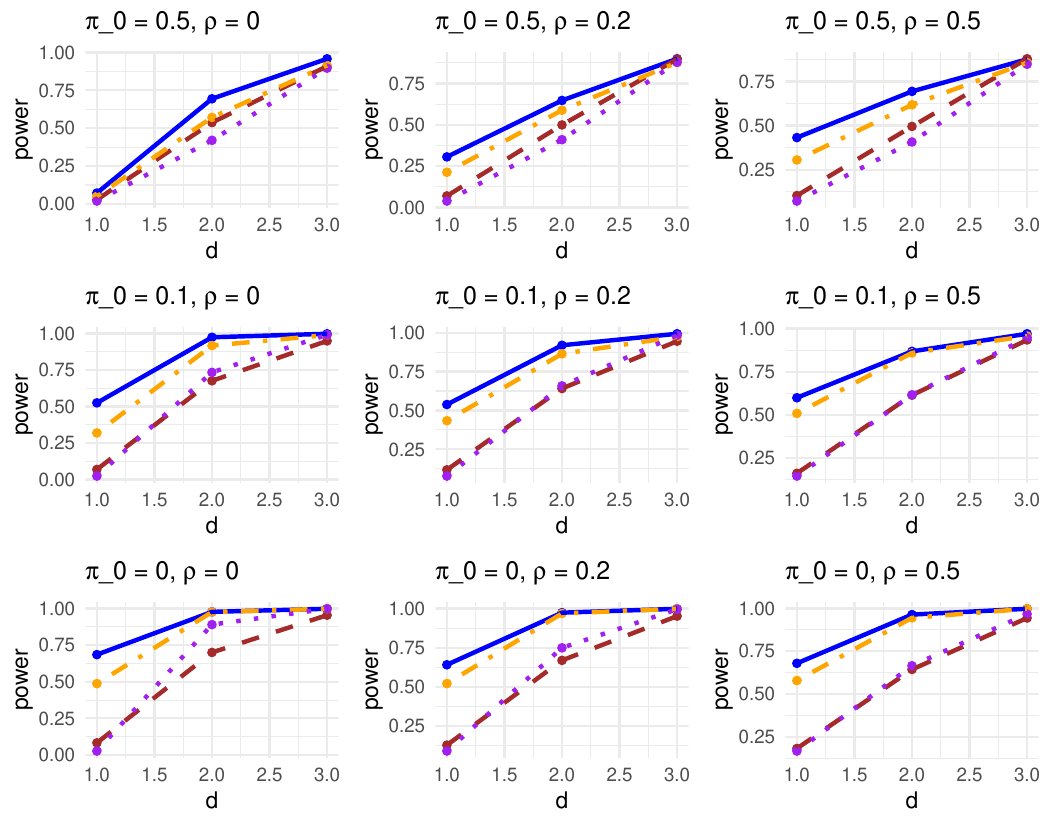}
  \caption{ \emph{The power of our novel method from Section \ref{seckorn} (solid lines) versus Benjamini-Hochberg (dashed), Romano-Wolf (dotted) and the method from \citt{hemerik2024flexible} (dash-dotted) as depending on  the fraction $\pi_0$ of true hypotheses, the correlation $\rho$ between the test statistics and the effect size $d$. The number of hypotheses was $m=500$. Each estimate is based on $5\cdot 10^3$ simulations.}
} \label{fig:pow}
\end{figure}

\subsection{Data analysis} \label{secdata}
To illustrate the methods, we applied them  to a dataset on  residential property sales in Ames, Iowa
from 2006 to 2010 \citp{cock2011ames}. We used the cleaned data which are available through the R package \verb|AmesHousing| \citp{ameshousing}. The dataset contains 2930 rows, which correspond to sold individual properties. The data  contain the sale price and various variables that might affect it, such as the lot area, the number of bathrooms above grade, the build year and the roof style.

We removed variables that were near-constant and standardized all quantitative variables, including the sale price. We fitted a standard least squares linear regression model, where sale price was regressed on all other variables, including the categorical variables. For all 31 quantitative covariates we considered the coefficient estimates $T_1,...,T_{31}$, say. 
We did not consider the coefficients of the dummy variables (representing categories), since many categories were very rare.

It is well-known that the coefficient estimates are asymptotically multivariate normal, so that Assumption \ref{asssym} is approximately satisfied.
For each $1\leq j \leq 31 $ we considered the one-sided null hypothesis 
$$H_j:\mu_j\leq -0.1$$ 
versus the complementary alternative, where $\mu_j=\mathbb{E}(T_j)$, which is the true value of the $j$-th coefficient if the linear model is correct.  Such hypotheses are of interest if we want to show that many of the covariates do not have a strong negative effect on the sales price. In this context we would expect that, since many of the quantitative variables relate in a positive way to the size of the residential property. 
We took $\gamma=0.1$ and applied the method from Section \ref{seckorn}, based on the  FDP estimator $\tilde{FDP}$  for one-sided tests from Section \ref{secnovel}. This led to rejection of 30 of the 31 hypotheses. This means that with $50\%$ confidence, we know that at least $(1-\gamma)\cdot 30 =27$ of these 30 coefficients are truly larger than $-0.1$. Note that a rejection here has the same meaning as when we would  control the FDR, except that we control the median instead of the mean of the FDP.
The only hypothesis that was not rejected, corresponds to the variable indicating the amount of unfinished basement surface. This was the only variable with a coefficient below $-0.1$. Thus, we see that all hypotheses with test statistics above $-0.1$ were rejected.

For comparison, we also applied the FDR method of Benjamini and Hochberg \citp{benjamini1995controlling} to test the hypotheses $H_j:\mu_j\leq -0.1$. The target FDR (often called $\alpha$) was also taken to be $\gamma=0.1$. As input for Benjamini-Hochberg, we used right-sided p-values based on the t-statistics $(\hat{\mu}_j-(-0.1))/{s.e.}$
 corresponding to these hypotheses. Benjamini-Hochberg rejected 29 out of 31 hypotheses, so it failed to reject one hypothesis with test statistic above $-0.1$.

It is interesting to see what happens when we reduce the sample size to $n=100$. Indeed, the p-values which e.g. Benjamini-Hochberg uses as input then tend to be less small.
  The mFDP method, on the other hand, is not based on the p-values but purely on the coefficient estimates. Of course, the power of the mFDP method is also expected to suffer somewhat---since the coefficient estimates will be less accurate.
  To reduce the sample size to $n=100$, we randomly sampled 100 rows from the dataset, and used these as the new data. We re-fitted the model, using only the quantitative predictors, to avoid having more than $n$ covariates. This time, among the 31 estimated coefficients, four were below $-0.1$.
    Our method rejected 27 out of 31 hypotheses, meaning that in this example, we were able to reject all hypotheses with estimated coefficients above $-0.1$.  
Benjamini-Hochberg rejected only 16 out of 31 hypotheses. Bonferroni (with $\alpha=0.1$) rejected 7 out of 31 hypotheses.

\section{Discussion}

In many situations, test statistics are  symmetric about their means $\mu_j$, either for finite samples or asymptotically. 
In this paper, we have taken this property as the main assumption.
If we expect a priori that  the $\mu_j$ tend to  lie in the directions or intervals corresponding to the alternative hypotheses, then our procedures  exploit this  and gain power compared to existing methods.
 At the same time, the methods are valid  regardless of whether these a priori expectations are correct.

Because of their combinatorical nature, nonparametric approaches such as ours   lend themselves better to controlling the median of the FDP  \citp{romano2007control,hemerik2018false, hemerik2019permutation} than to controlling the mean of the FDP, i.e., the FDR. However, FDR control is definitely an attractive criterion. For example, FDR control always implies \emph{weak} FWER control. This means that if all hypotheses are true, then the probability of at least one incorrect rejection will be controlled. On the other hand, in the settings that we are interested here, most hypotheses tend to be false, so weak FWER control is less of an advantage. Besides, we could guarantee weak FWER control as well by first performing a global test, which tests whether all $m$ hypotheses are true.
 Another relevant point though is that the distribution of the FDP tends to be skewed to the right, so that naturally the mFDP tends to be smaller than the FDR. 
In conclusion,  mFDP control should only be preferred over FDR control when there are specific reasons for this, such as the availability of an mFDP method with much better power or with fewer assumptions. For example, in the settings that we are interested in in this paper, FDR methods with proven finite-sample validity tend to be very conservative, which can be a reason to use our mFDP approach instead. 

The way in which we estimate  the FDP  based on a symmetry property is related to the knockoffs procedures from  \citt{barber2015controlling}.   In many respects those methods are very different and the underlying setting is different.
However, the similarity is that the knockoffs procedures are also based on counting how many test statistics are above  $t$ and below $-t$, for different $t\geq 0$. 
A difference between their approach and ours is that null  test statistics based on knockoffs are mutually uncorrelated, whereas here they are potentially correlated. 

Our multiple testing approach is clearly nonparametric (or semiparametric), since it requires no assumptions about the distributional shapes of the test statistics, except symmetry.
What makes our methods different from other nonparametric procedures---e.g. permutation methods such as \citt{hemerik2018false}, \citt{hemerik2019permutation} and \citt{blain2022notip}---is that we do not directly use symmetries of the data, but only of the resulting test statistics. This would seem very crude, and e.g. using permutations---if possible---would seem more sophisticated. However,  when most of the hypotheses are false, our approach tends to lead to  very low bounds for $V$ and consequently high power. 
Thus, our approach is very \emph{adaptive}, in the sense that it adapts very well to the amount of signal in the data. 
Additionally, when implemented with care, all methods in this paper have a computation time that is linear in $m$, after sorting the test statistics.

\setlength{\bibsep}{3pt plus 0.3ex}  
\def\bibfont{\small}  

\bibliographystyle{biblstyle}
\bibliography{bibliography}

\appendix

\section{Testing a single hypothesis: examples} \label{secsingle}
Nonparametric and semiparametric tests can be applied in many situations \citp{winkler2014permutation, canay2017randomization, hemerik2020robust,  berrett2020conditional, hemerik2021permutation, hemerik2021another, dobriban2022consistency, liu2022fast, zhang2023randomization}.  As with parametric tests, nonparametric tests are only exact for finite samples when the model is relatively simple. Here will give examples of tests that are valid for finite samples.
Another example of an exact nonparametric test is a  test
of independence of two continuous variables as in \citt{neyman1942basic, diciccio2017robust, kim2022local}. Semiparametric tests can also be used for inference in generalized linear models, for example, but then they are only asymptotically exact of course \citp{de2025inference, de2025permutation}.

\subsection{Location model} \label{secloc}
Often we want to show that a population mean $\theta\in \mathbb{R}$ lies above some value $\delta\in  \mathbb{R}$. In this case, the null hypothesis is $H_0:\theta \leq \delta$ and the alternative  $H_a:\theta > \delta$. 
To test $H_0$, we might use a t-test. However, an advantage of a nonparametric approach is that it requires fewer assumptions. Moreover, nonparametric tests can be combined with powerful permutation-based multiple testing methods \citp{westfall1993resampling,westfall2008multiple,pesarin2010permutation,hemerik2018false,hemerik2019permutation,
blain2022notip,andreella2023permutation}. 
\citt[][\S21]{fisher1935} introduces a nonparameric test based on sign-flipping. This test does not require normality, but instead requires the milder assumption that under $H_0$, the observations are symmetric about their mean. In particular, there is no homoscedastity assumption.

The model is the following. The data is a vector $X$ taking values in $\mathbb{R}^n$, where $n\geq 1$ is the sample size. The $n$ observations are assumed to be independent and to have the same mean $\theta$. Further,  they are assumed to be symmetrically distributed about their mean under $H_0$. The observations can have different variances and distributions.

Note that without loss generality, we can assume that the null hypothesis is $H_0: \theta\leq 0$. To test $H_0$, one can  consider the group of all sign-flipping transformations, i.e., all transformations $g: \mathbb{R}^n \rightarrow \mathbb{R}^n$ of the form $x\mapsto F_gx$, where $F_g$ is a diagonal matrix with diagonal elements in $\{-1,1\}$. There are $2^n$ such transformations $g$ and they form a group, say $\G$,  in the algebraic sense \citp{hemerik2018exact,koning2024more}. We refer to this group as the \emph{sign-flipping group}. For every transformed data vector $F_gX$, the test computes the  test statistic $T_g(X):=n^{-1/2}1_n'F_gX$, where $1_n$ is the $n$-vector of ones. For the original data, the test statistic is $$T_{id}(X)=n^{-1/2}1_n'I_nX =    n^{-1/2}\sum_{i=1}^nX_i,$$ where $I_n$ is the $n\times n$ identity matrix.  Given $\alpha\in(0,1)$, the test rejects $H_0: \theta\leq 0$ if and only if $T_{id}(X)$ is larger than $(1-\alpha)100\%$ of the other statistics. To make this formal, let $T^{(1)}(X)\leq ...\leq T^{(|\G|)}(X)$ be the $|\G|=2^n$ sorted test statistics. 
Write $T^{(1-\alpha)}(X) := T^{\lceil (1-\alpha)|\G| \rceil }(X)$,  where $\lceil c \rceil$ denotes the smallest integer than is at least as large as $c$. 
Then the test rejects if and only if
\begin{equation} \label{fliptest}
T_{id}(X)>T^{ (1-\alpha) }(X),
\end{equation}
It is well known that this test has type I error rate at most $\alpha$ \citp{arboretti2021unified}, although most literature assumes a point null hypothesis.
\begin{theorem} \label{thmflip}
Consider $X$ as above and $H_0: \theta\leq 0$.
The test that rejects when inequality  \eqref{fliptest} holds, satisfies $\mathbb{P}_{H_0}(\text{reject } H_0)\leq \alpha$. Further, if $\theta = 0$ and $\alpha$ is a multiple of $|G|^{-1}$ and $X$ is continuous, then  $\mathbb{P}_{H_0}(\text{reject } H_0) =\alpha$. 
\end{theorem}

\subsection{Comparing two groups} \label{sec2gr}
Another nonparametric model relevant to this paper is the following one, often used for comparing two populations by exact permutation testing.
Consider two samples $Z$ and $Y$, which both take values in $\mathbb{R}^{n}$. Here $Z$ usually represents a sample from one population and $Y$ a sample from another population.
The assumption that both samples are equally large is made for convenience. We write the whole dataset as $X=(Z',Y')'$. The permutation test employs permutation maps $g:\mathbb{R}^{2n}\rightarrow \mathbb{R}^{2n}$. A permutation map is a  transformation of the form $(x_1,...,x_{2n})'\mapsto (x_{\pi(1)},...,x_{\pi({2n})})'$, where $(\pi(1),...,\pi({2n}))$ is a permutation of $(1,...,2n)$. Note that there are $(2n)!$ such permutation maps. 

As a concrete example, suppose that all $2n$ observations are i.i.d., except that the observations in $Z$ have mean $\theta^1$ and the observations in $Y$ have mean $\theta^2$. Then we can test the null hypothesis $H_0: \theta \leq \delta$, where $\theta= \theta^1-\theta^2$. Without loss of generality we can assume that $\delta=0$.
To test $H_0$, we compute the test statistic
$$T(X)=T_{id}(X)= T\{(Z',Y')'\}:=n^{1/2}(\bar{Z} - \bar{Y}),$$
where $\bar{Z}$ denotes the mean of the entries of $Z$ and likewise for $Y$.
Letting $\G$ be the group of $(2n)!$ permutation maps, for every $g\in\G$
we compute the test statistic $T_g(X)= T(g(X))$.
As in section \ref{secloc}, we define $T^{(1-\alpha)}(X)$ to be the $(1-\alpha)$-quantile of these test statistics and reject $H_0$ if and only if $T(X)>T^{(1-\alpha)}(X)$. 
As a sidenote, it can be seen that many permutations lead to the same test statistic, so that in fact we only need to use $\binom{2n}{n}$  permutations \citp{hemerik2018exact}, which speeds up the computation.

It is well-known that this test controls the type I error rate \citp{pesarin2015some}, although this is usually proved for the case of a point null hypothesis \citp{hoeffding1952large,lehmann2022testing,hemerik2018exact}. 

\begin{theorem} \label{thm2gr}
Consider $X=(Z',Y')'$ as above.
Suppose that all $2n$ observations are i.i.d., except that the observations in $Z$ have mean $\theta^1$ and the observations in $Y$ have mean $\theta^2$. Consider  $H_0: \theta\leq 0$, where $\theta= \theta^1-\theta^2$. 
Consider the test that rejects when inequality  \eqref{fliptest} holds, which is now understood to refer to the current setting.
This test satisfies $\mathbb{P}_{H_0}(\text{reject } H_0)\leq \alpha$. Further, in case $\theta = 0$ and $\alpha$ is a multiple of $\binom{2n}{n}$ and $X$ is continuous, then  $\mathbb{P}_{H_0}(\text{reject } H_0) =\alpha$. 
\end{theorem}

\begin{remark}
For what follows later (e.g. Section \ref{secmainas}), it is important to realize that in the setting of Theorem \ref{thm2gr}, the test statistic $T(X) = n^{1/2}(\bar{Z}-\bar{Y})$ is symmetric about its mean $n^{1/2}\theta$, even if the $2n$ variables are not symmetric.
To see this, note  that $(\bar{Z}-\bar{Y})$ is symmetric about $\theta$, which follows from
$$ (\bar{Z}-\bar{Y}) - \theta = (\bar{Z}- \theta^1) -(\bar{Y} - \theta^2) \,{\buildrel d \over =}\, (\bar{Y}- \theta^2) -(\bar{Z} - \theta^1) =\theta - (\bar{Z}-\bar{Y}),$$ where we used that $(\bar{Z}- \theta^1)\,{\buildrel d \over =}\,(\bar{Y} - \theta^2)$ by assumption, as well as independence of $\bar{Z}$ and $\bar{Y}$. 
\end{remark}

\section{Examples where Assumption \ref{asssym} is satisfied for finite samples}  \label{secexvalidass}
We now provide two examples of settings where the test statistics are not normally distributed and  Assumption \ref{asssym} is eactly satisfied for finite samples.

\begin{example}[Nonparametric location model]  \label{exltest}
In Section \ref{secloc} we considered a non-parametric location model and a hypothesis test. Now consider not one such test, but $m$ such hypotheses and tests. Now $X$ is not an $n$-vector but an $n$-by-$m$ matrix
$X=(X^1,...,X^m)$, where $X^j$ is an $n$-vector whose entries have mean $\theta_j$, $1\leq j \leq m$.  For every $1\leq j \leq m$ we now consider the null hypothesis $H_j:\theta_j\leq \delta$ and define $T_j = n^{-1}1_n'X^j$. 
 Note that $\mu_j=\mathbb{E}(T_j)= \theta_j$ and $\N = \{1\leq j \leq m: \theta_j\leq \delta\}$.
Let $X^{\N}$ be the $n$-by-$|\N|$ submatrix $(X^j:j\in\N)$ and write  $X^{\N}_i$ for its $i$-th row.
Assume that the rows of  $X^{\N}$ are independent of each other. Within each row, there may be dependence.
Let $\theta^{\N}:=(\theta_j:j\in\N)$ and assume that for every $1\leq i \leq n$,  the rows of  $X^{\N} - \theta^{\N} $ are symmetric, i.e. that   $X^{\N}_i-\theta^{\N}  \,{\buildrel d \over =}\, -(X^{\N}_i-\theta^{\N} )$.

Note that the entries $(T_j-\mu_j : j\in \N)$ have mean $0$ and satisfy
\begin{align*}
 (T_j-\mu_j : j\in \N) = & \big\{(n^{-1}1_n'X^j -\theta_j) : j\in \N\big\}  \\
=& \big\{n^{-1}(1_n'X^j-n\theta_j)  : j\in \N\big\} \\
=& \big\{n^{-1}\sum_{1\leq i \leq n}(X_i^j-\theta_j)  : j\in \N\big\}.
\end{align*}
Since the rows of  $X^{\N}$ are symmetric and independent of each other, this is equal in distribution to
\begin{align*}
&  \big\{n^{-1}\sum_{1\leq i \leq n}(-(X_i^j-\theta_j))  : j\in \N\big\} =\\
&  \big\{  n^{-1}(-(1_n'X^j-n\theta_j))  : j\in \N\big\}=\\
& (-(T_j-\mu_j) : j\in \N),
\end{align*}
so that Assumption \ref{asssym} is satisfied.
\end{example}

In the example above we assumed that the observations are symmetric. In the example below,  we do not make such an assumption, yet still the test statistics satisfy Assumption \ref{asssym}.

\begin{example}[Comparing two groups]  \label{ex2gr}
We now extend the setting of Section  \ref{sec2gr} to the case of multiple hypotheses; here we will also see that Assumption \ref{asssym} is satisfied. Now $X$ is a $2n$-by-$m$ matrix $X=(X^1,...,X^m)$, where for each $1\leq j \leq m$, 
$X^j = (Z^{j'},Y^{j'})'$ is a $2n$-dimensional column vector, where $Z^{j}$ is the j-th column of an $n$-by-$m$ matrix $Z$ and  $Y^{j}$ is the j-th column of an $n$-by-$m$ matrix $Y$.
For $1\leq j\leq m$, the entries of $Z^{j}$ have mean $\theta^1_j$ and the entries  of $Y^{j}$ have mean $\theta^2_j$.
 For every $1\leq j \leq m$, consider the null hypothesis $H_j:\theta_j \leq \delta$, where $\theta_j=\theta^1_j-\theta^2_j$.
 For  each $1\leq j \leq m$, let
$T_j = n^{1/2}(\bar{Z}^j-\bar{Y}^j)$, where $\bar{Z}^j$ is the mean of the entries of ${Z}^j$ and likewise for  ${Y}^j$.
Note that $\mu_j:=\mathbb{E}(T_j)=n^{1/2}(\theta^1_j - \theta^2_j)     =  n^{1/2} \theta_j$.
Let $X^{\N}$ be the $n$-by-$|\N|$ submatrix $(X_j:j\in\N)$. 
Assume that the rows of  $X^{\N}$ are independent of each other and each have the same joint distribution, apart from a mean shift so that the columns of $Z$ and $Y$ have the means specified above.
Within each row, there may be dependence. 
Thus, the matrices $Z^{\N}$ and $Y^{\N}$ have the same joint distribution, except that the columns of $Z^{\N}$ can have different means than the columns of $Y^{\N}$.

 Note that the entries $(T_j-\mu_j : j\in \N)$ have mean $0$ and satisfy

 \begin{align*}
&      (T_j-\mu_j : j\in \N) =                 \\
&    \big\{n^{1/2}(\bar{Z}^j-\bar{Y}^j)   -n^{1/2}\theta_j : j\in \N\big\} =                    \\
&    \Big\{n^{1/2}\big((\bar{Z}^j- \theta_j^1)-(\bar{Y}^j- \theta_j^2)\big)  : j\in \N\Big\}  =                   \\
&  \Big\{n^{1/2}n^{-1}\sum_{1\leq i \leq n} \big[({Z}_i^j-  \theta_j   )-({Y}_i^j- \theta_j^2)\big]  : j\in \N\Big\}=                                     
\end{align*}
\begin{equation} \label{eq:ass2twogroup}
n^{1/2}n^{-1}\sum_{1\leq i \leq n} \Big\{ ({Z}_i^j-  \theta_j   )-({Y}_i^j- \theta_j^2)  : j\in \N\Big\}.
\end{equation}
Note that for every $1\leq i \leq n$, 
 \begin{align*}
&  \big\{({Z}_i^j-  \theta_j^1   )-({Y}_i^j- \theta_j^2): j\in \N\big\}\,{\buildrel d \over =}\, \\ 
&\big\{({Y}_i^j- \theta_j^2)-({Z}_i^j-  \theta_j^1   ): j\in \N\big\},    
\end{align*}
where we used that
$$(Z^j_i -\theta^1_j:j\in\N )  \,{\buildrel d \over =}\, (Y^j_i- \theta^2_j:j\in\N )$$
combined with the fact that the rows of $X$ are mutually independent.

Hence, \eqref{eq:ass2twogroup} is equal in distribution to
 \begin{align*}
&  n^{1/2}n^{-1}\sum_{1\leq i \leq n} \Big\{ \big(({Y}_i^j- \theta_j^2) -({Z}_i^j-  \theta_j   )\big)  : j\in \N\Big\}=     \\
&  \Big\{n^{1/2}\big[(\bar{Y}^j- \theta_j^2)  -   (\bar{Z}^j- \theta_j^1)\big]  : j\in \N\Big\}=     \\
&     \Big\{n^{1/2}(\bar{Y}^j-\bar{Z}^j)  + n^{1/2}\theta_j : j\in \N\Big\} =  \\
&    (-(T_j-\mu_j) : j\in \N),   
\end{align*}
 so that Assumption \ref{asssym} is satisfied.

\end{example}

\begin{remark}
Note that in Example \ref{ex2gr}, we assumed that the rows of the data matrix $X$ were independent of each other. While this is a common assumption, it is in fact not needed at all for Assumption \ref{asssym}  to be satisfied. Indeed, it suffices to assume that the matrices $Z^{\N}$ and $Y^{\N}$ have the same joint  distribution after removing the means of their entries. The reason is that Assumption \ref{asssym} is quite mild. In particular, we do not need invariance under all permutations like SAM does.
\end{remark}

\begin{remark}
In Example \ref{ex2gr}, the hypotheses  were $H_j: \theta_j\leq \delta$. This can straightforwardly be generalized to the more general collection of hypotheses $H_j': \theta_j\leq \delta_j$, where $\delta_1,...,\delta_m \in \mathbb{R}$. Indeed, for every $1\leq j \leq m$ we can shift, i.e. translate, the observations in $Z^j$ by $\delta-\delta_j$. Then $H_j$ holds for the shifted data if and only if $H_j'$ holds for the original data. Likewise, we can generalize to the situation where some of the hypotheses are of the form $H_j'': \theta_j\geq\delta_j$, with a left-sided alternative. Indeed, if we multiply  the corresponding columns of $X$ by $-1$ and then suitably shift the data, the problem becomes that of Example \ref{ex2gr}.
\end{remark}

\section{Theoretical comparison of the estimator $\tilde{V}$ with SAM} \label{seccompSAM}
 In this section, we discuss a notable relationship between the bound  $\tilde{V}$ defined in Section \ref{secnovel}  and  the existing competitor SAM, which was developed in      \citt{hemerik2018false}     and was summarized in Section \ref{secexm}. 
We show that in situations where SAM can be used, SAM is sometimes equivalent to the novel method, if in SAM one takes $\alpha=0.5$ and uses only two permutations. Using only two permutations might seem a bad idea power-wise \citp{chung1958randomization,marriott1979barnard}. However, if we perform one-sided tests and most hypotheses are false and $\alpha=0.5$, then  this often gives better performance than when many permutations are used. In other words, by using two permutations we often obtain a  lower  median unbiased estimator of $V$ than when many permutations are used. Note that this is only true in certain one-sided testing settings where most of the hypotheses are false.

Firstly, we show that SAM with only two permutations is often equivalent to the novel method. Note that in some one-sided permutation-type tests, e.g. those of Sections \ref{secloc} and \ref{sec2gr}, the permutation distribution is symmetric about 0, meaning that  for every $g\in \G$ there is a $g'\in \G$ such that $T(g(X)) = -T(g'(X))$. Note that the equality here is not just an equality in distribution.
Thus, if we make a histogram of the test statistics---including the original test statistic---then this histogram will be symmetric about 0---regardless of whether $H_0$ is true. 
When we test multiple (say $m$) hypotheses with a permutation method, then for every $1\leq j \leq m$ and for every $g\in \G$, we compute a statistic $T_j(g(X))$ (see e.g. Section \ref{secexm}). If the test statistics are one-sided, they sometimes have the analogous  symmetry property: for every  $g\in \G$, there is a $g'\in \G$ such that 
\begin{equation} \label{symtstats}
\big\{T_1(g(X)),...,T_m(g(X))\big\} = - \big\{T_1(g'(X)),...,T_m(g'(X))\big\}.
\end{equation}

Now, recall from Section \ref{secexm} that for $\alpha=0.5$, SAM computes  $\bar{V}(t)$ as the median of the numbers $\R(g(X),t)$, $g\in \G$, i.e, the median of the numbers
$$|\{1\leq j \leq m: T_j(g(X))>t|,$$
$g\in G$.
But due to the symmetry property \eqref{symtstats}, this equals the median of the numbers 
$$|\{1\leq j \leq m: T_j(g(X))< - t|,$$
$g\in G$.
Suppose now that rather than a large group $\G$, we use a subgroup of only 2 transformations. In case of sign-flipping (Section \ref{secloc}), we consider the subgroup $\{id, h\}$, where $id$ is the identity map $(x_1,...,x_n)'\mapsto (x_1,...,x_n)'$ and $h$ is the map  $(x_1,...,x_n)'\mapsto- (x_1,...,x_n)'$. In case of permutation tests for comparing two equally sized groups, we take $h$ to be a permutation that swaps all cases with the controls---which typically has the same effect on the test statistic as multiplying the data $X$ by -1 (this is the case in e.g. Example \ref{ex2gr}).
If we use this tiny  subgroup,  then we write $\bar{V}_2$, saving the notation $\bar{V}$ for the bound based on the whole group $\G$ (see Section \ref{secexm}). 
Note that $\bar{V}_2(t)$ is simply
$$\bar{V}_2(t) = R(h(X),t)\wedge R(X,t)\text{, where}$$   
$$   R(h(X),t) =    |\{1\leq j \leq m: T_j(-X)> t| =  |\{1\leq j \leq m: T_j(X)<-t|.$$
Thus,  $\bar{V}_2(t)$ is simply  the bound $R(t)\wedge R^-(t)= \tilde{V}(t)$ from Section \ref{secnovel}, if we take $\delta_1=...=\delta_m=0$. Thus, SAM with $\alpha=0.5$ and two permutations, is in these cases equivalent to the novel method with  $\delta_1=...=\delta_m=0$.

As confirmed with simulations, if $\alpha=0.5$,  then $\bar{V}_2$ often performs better than $\bar{V}$, i.e., $\bar{V}_2$ if often  lower on average. Again, this is only true in certain one-sided testing settings where most hypotheses are false. 
In these settings, in SAM it often better to use merely two transformations than the whole group. 
We do not obtain a uniform improvement though. Indeed, we only have $\bar{V}_2\leq \bar{V}$ if  $R^-(X)  \leq R^{(0.5|\G|)}(X)$, which does not hold in general as can easily be checked with an example. We provide a simple example where $\tilde{V} = \bar{V}_2>\bar{V}$ below. It is also easy to find examples where $\tilde{V}=\bar{V}_2<\bar{V}$.

\begin{example}
Suppose the data  are $X=(X^1,X^2)$, where  $X^j\in \mathbb{R}^n$ for $j\in\{1,2\}$ and $n=3$. 
Assume that for every $j\in\{1,2\}$, all observations in column $j$
have mean $\theta^j\in \mathbb{R}$, say. 
 There are $m=2$ null hypotheses $H_j: \theta^j \leq 0$, $j\in\{1,2\}$. 
Let the test statistics be  $T_j(X) = \sum_{i=1}^{n} X^j_i$, $j\in\{1,2\}$, and suppose Assumption \ref{asssym} holds.
Pick $t=1$.
Suppose the observed values are $X^1_1=...=X^1_n=0.34$ and  $X^2_1=...=X^2_n= - 0.34$. 
Thus, $$R(X) = |\{1\leq j \leq m: T(X)>t\}| = |\{1\}|=1,$$
$$\tilde{V} = \bar{V}_2 = |\{1\leq j \leq m: T(X)<-t\}| = |\{2\}|=1.$$
The group $\G$ of sign-flipping transformations has cardinality $2^n=8$.
To obtain $\bar{V}$, we compute the median of the numbers of rejections $R(g(X))$. For all $g\in\G$ except $id$ and the map $X\mapsto -X$, we find  $R(g(X))=0$, so that the median is $0$ and $\bar{V}=0$, which is strictly smaller than  $\tilde{V}= \bar{V}_2$. 
\end{example}

\section{Closed testing and admissibility of $\tilde{V}'$} \label{secctad}   
Consider the one-sided testing setting from Section \ref{secnovel}.
In this section, we link the bounds $\tilde{V}$ and $\tilde{V}'$ to the theory of \emph{closed testing} and prove that the bound $\tilde{V}'$ is admissible, i.e., that it cannot be uniformly improved.

\subsection{General construction of FDP bounds using closed testing} \label{secbasicsct}
We start by reviewing what closed testing is and how it can be used to obtain bounds $B$ of the form $\mathbb{P}(V\leq B)\geq 1-\alpha$, where $\alpha\in(0,1)$ is fixed in advance. 
The closed testing principle goes back to \citt{marcus1976closed} 
and can be used to construct multiple testing procedures that control the family-wise error rate \citp{sonnemann2008general,  romano2011consonance}. \citt{goeman2011multiple} show that such  procedures  can be extended to provide confidence bounds for the numbers of true hypotheses in all sets of hypotheses simultaneously \citp[an equivalent approach is in][]{genovese2006exceedance}. 

Let $\C$ be the collection of all nonempty subsets of $\{1,...,m\}$.
For every  $I\in \C$ consider the intersection hypothesis
$H_I=\cap_{i\in I} H_i$. This is the hypothesis that all $H_i$ with $i\in I$ are true.
For every $I\in \C$, consider some \emph{local test} $\phi_I$, which is $1$ if $H_I$ is rejected and $0$ otherwise. Assume  the test $\phi_\N$ is valid in the sense that it has level at most $\alpha$, i.e.,  $\mathbb{P}(\phi_\N\geq 1)\leq \alpha$. Since we do not know $\N$, this effectively means that we require all the local tests to be valid.
A closed testing procedure generally rejects all intersection hypotheses $H_I$ with $I\in \X$, where 
$$\X:=\{I \in \C: \phi_J=1 \text{ for all } I\subseteq  J\subseteq  \{1,...,m\}  \}.$$
 It is well-known that this procedure controls the familywise error rate \citp{marcus1976closed}, i.e., with probability at least $1-\alpha$ there are 0 type I errors.
In \citt{goeman2011multiple} it is shown that we can also use the set $\X$ to provide a  $(1-\alpha)$-confidence upper bound for the number of true hypotheses in any $I\in\C$.
They show that
$$t_{\alpha}(I):=\max\{|J|: J \subseteq I \text{ and }   J\not\in \X\}.$$
is a  $(1-\alpha)$-confidence upper bound for $|\N\cap I|$. In fact, they show that the bounds $t_{\alpha}(I)$ are valid simultaneously over all $I\in\C:$
\begin{equation} \label{GSresult}
   \mathbb{P} \Bigg[\bigcap_{I\in\C}\Big\{     |\N\cap I| \leq t_{\alpha}(I)   \Big\} \Bigg]\geq 1-\alpha.
   \end{equation}
The proof is short: 
with probability at least $1-\alpha$, $H_{\N}$ is not rejected by its local test. In case $H_{\N}$ is not rejected by its local test,  for every $I \in \C$ it holds that $|\N\cap I|\not\in \X$, so that $|\N\cap I| \leq t_{\alpha}(I)$, i.e., $t_{\alpha}(I)$ is a valid upper bound for the number of true hypotheses in $I$. This proves inequality \eqref{GSresult}.
A different method,  formulated in \citt{genovese2006exceedance}, leads to the same bounds. 
The two approaches were compared in  the supplementary material of \citt{hemerik2019permutation}.

\subsection{$\tilde{V}$ coincides  with a closed testing method} \label{secvtisctp}
We first consider the bound $\tilde{V}$. We will define local tests and, as explained in Section \ref{secbasicsct}, use closed testing to obtain bounds. Then we show that these bounds  are equivalent with $\tilde{V}$, so that $\tilde{V}$ coincides with a closed testing procedure.

Recall the definition of  $\C$ from Section \ref{secbasicsct}.
For every $I\in \C$, we will  define a corresponding local test $\lt_I$. First define 
$$R_I(t) := |I\cap \R(t)| =  \{j\in I: T_j>\delta_j+t\},$$
$$R^-_I(t) := |I\cap \R^-(t)| =  \{j\in I: T_i<\delta_j-t\}.$$
We will often refrain from writing ``$(t)$'' for brevity.
For every $I\in \C$,  we define the local test 
\begin{equation} \label{deflocalbasic}
\lt_I =\mathbbm{1}(R_I>R^-_I),
\end{equation}
where $\mathbbm{1}(\cdot)$ is the indicator function. 

\begin{proposition} \label{ltval}
Suppose Assumption \ref{asssym} holds.
Then for every  $I\in\C$, the local test  $\lt_I$ is valid for $\alpha=0.5$, in the sense that if   $H_I$ is true, then  $\mathbb{P}(\lt_I=1)\leq 0.5.$
\end{proposition}

Proofs are in section \ref{appproofs}.
We now investigate which $I$ are in $\X$, i.e, which intersection hypotheses are rejected by the closed testing procedure.
\begin{proposition} \label{whichinX}
Consider the closed testing procedure based on the local tests defined by equation \eqref{deflocalbasic}.
For any $I\in \C$, we have  $I\in\X$ if and only if  $R_I>R^-$, where we recall that $R^-=R^-_{\{1,...,m\}}$.
\end{proposition}

Next we investigate what bounds the above closed testing procedure provides. For $I\in \C$, note that
$$t_{\alpha}(I)=\max\{|J|: J \subseteq I \text{ and }   J\not\in \X\}= $$
\begin{equation} \label{boundnotadm}
\max\{|J|: J \subseteq I \text{ and }  R_J\leq R^- \}.  
\end{equation}

In particular, we are interested in $t_{\alpha}(\R(t))=t_{\alpha}(\R)$, since this provides a bound for $| \N\cap \R(t)| =V(t)$, which is the number of incorrect rejections for our method defined in Section \ref{secnovel}.
Note that 
$$t_{\alpha}(\R(t)) = \max\{|J|: J \subseteq \R \text{ and }   R_J\leq R^- \} =$$
$$   \max\{|J|: J \subseteq \R \text{ and }  |J|\leq R^- \} = R\wedge R^- =\tilde{V}(t).$$
Thus, $\tilde{V}(t)$ coincides with the $50\%$-confidence bound for $|\N\cap\R|$ produced by this closed testing procedure.
Now consider any other $I\in \C$.
Expression \eqref{boundnotadm} can  be rewritten as
$$ |I|\wedge \big(   |\{j\in I: T_j\leq \delta_j+t  \}| + R^-   \big) =  |I| \wedge \big(  |I\cap \R^c| +R^-    \big),$$
where $\R^c:=\{1,...,m\}\setminus \R$. Thus, the  $50\%$-confidence bound for the number of true hypotheses in $I$, $|\N\cap I|$, is simply the bound  $R^- $ for $|\N\cap \R|$ plus the number of elements in $I$ that are not in $\R$. But this bound follows trivially from the bound  $\tilde{V}=R^-\wedge R$ for $|\N\cap R|$, without looking at the data.  Thus, all the information contained in the bounds $t_{\alpha}(I)$, $I\in \C$, is already contained in our single bound $\tilde{V}$. Thus, the bound $\tilde{V}$ is completely equivalent to the bounds from the closed testing procedure. As we showed, $\tilde{V}$ is not admissible---since it is improved by $\tilde{V}'$---and hence this closed testing procedure is not admissible.

\subsection{$\tilde{V}'$ is admissible} \label{secvtildepadm}

Next, we will show that $\tilde{V}'$ also coincides with a closed testing procedure and that that procedure is admissible. We will then  deduce that  $\tilde{V}'$ is admissible. We start by defining the closed testing procedure, which we do by defining its local tests.
For every $I\in \C$,  we define the local test 
\begin{equation} \label{defadmlt}
    \lt'_I= 
\begin{cases}
   1  & \text{if } R_I(t) >  R_I^-(t) \\
   b  & \text{if } R_I(t) =  R_I^-(t) \\
  0            & \text{otherwise,}
\end{cases}
\end{equation}
i.e., $$  \lt'_I =  \mathbbm{1}(R_I >  R_I^-)+ b    \mathbbm{1}(R_I =  R_I^-),$$
where we define $b$ to be a coin flip  satisfying $b=1$ when $d=-\infty$ and $b=0$ if $d=1$. Note that $\lt'_I$ is a valid local tests for $\alpha=0.5$, i.e., if $H_I$ is true then $\pr( \lt'_I=1)\leq 0.5. $

We first show that this test exhausts $\alpha=0.5$, under some distributions.
\begin{lemma}  \label{exhaust}
For every $I\in\C$,  there exists a distribution of $(T_1,...,T_m)$ for which 1. Assumption \ref{asssym}  is satisfied, 2. $H_I$ is true  and 3. 
$\mathbb{P}(\lt'_I=1)=0.5$.
\end{lemma}

 In many  situations, if a test exhausts $\alpha$ under some distributions, then it is admissible. 
 The lemma below states that this is also the case for the local tests $\lt'_I$,  $I\in\C$.

\begin{lemma} \label{ltIadm}   
Assume the support of $(T_1,...,T_m)$ is $\mathbb{R}^m$.
Then for every $I\in\C$, the test  $\lt'_I = \lt'_I((T_1,...,T_m))$ is admissible, i.e., there exists no test $\phi^*_I((T_1,...,T_m))$ taking values in $\{0,1\}$ such that:

\begin{enumerate}
\item the test $\phi^*$ is valid for $\alpha=0.5$, i.e.,   it generally holds that  $\mathbb{P}(\phi^*=1)\leq 0.5$ under Assumption \ref{asssym} if $H_I$ is true;
\item $\phi ^*_I$ always rejects when  $\lt'_I$ rejects, i.e., $ \phi ^*_I  \geq \lt'_I$;
\item for at least one distribution satisfying Assumption \ref{asssym}, $\mathbb{P}( \phi^*_I >\lt'_I )>0$, i.e., $\phi^*_I$ improves $\lt'_I$.
\end{enumerate}

\end{lemma}

 We now investigate which $I$ are in $\X$ for the closed testing procedure based on the local tests $\lt'_I$.
\begin{proposition} \label{propctp2}
Consider the closed testing procedure based on the local tests $\lt'_I$.
For any $I\in \C$, we have  $I\in\X$ if and only if  $R_I(t)>\tilde{V}'(t)$.
\end{proposition}

By Proposition \ref{propctp2}, the bounds that the closed testing procedure based on the  local tests $\lt'_J$ provides are 
\begin{equation} \label{boundadm}
\begin{aligned}
t_{\alpha}(I)&=\max\{|J|: J \subseteq I \text{ and }  J\not\in \X\} \\
&=\max\{|J|: J \subseteq I \text{ and }   R_J\leq \tilde{V}' \},
\end{aligned}
\end{equation}
$I\in \C$.
In particular, for $I=\R(t)$, which is the set we are interested in, we obtain
$$t_{\alpha}(\R(t)) = \tilde{V}',$$
analogously to Section  \ref{secvtisctp}.
Thus, $\tilde{V}'(t)$ concides with the $50\%$-confidence bound from the closed testing procedure based on local tests $\lt'_I$.
Now consider any other $I\in \C$.
The number \eqref{boundadm} equals
\begin{align*}
\max\{|J|: J \subseteq I \text{ and }   R_J\leq \tilde{V}'  \}  
&= |I|\wedge (  | \{j\in I: T_j\leq \delta_j+t  \}| + \tilde{V}'   ) \\
&=  |I| \wedge (  |I\cap \R^c| + \tilde{V}'  ).
\end{align*}
 Thus, the  $50\%$-confidence bound for the number of true hypotheses in $I$, $|\N\cap I|$, is simply the bound  $\tilde{V}' $ plus the number of elements in $I$ that are not in $\R$. This bound follows trivially from the bound  $\tilde{V}'$ for $\N\cap R$, without looking at the data.  Thus, all the information contained in the bounds $t_{\alpha}(I)$, $I\in \C$, is already contained in our single bound $\tilde{V}'$. 
 
 Earlier we found  that $\tilde{V}$ is equivalent to the closed testing procedure based on local tests $\lt(I)$. Now we have found that  $\tilde{V}'$ is equivalent to the closed testing procedure based on local tests $\lt'(I)$.
 An important difference is that unlike the local tests  $\lt_I$,  the local tests $\lt'_I$ are admissible. From Theorem 3 in \citt{goeman2021only} it follows that the corresponding closed testing procedure is admissible in the sense of  \citt{goeman2021only}.
Theorem \ref{thmadm} states that this means that the bound $\tilde{V}'$ is admissible. This is proved in section Section \ref{appproofs}.

\section{Equivalence testing: $\tilde{V}$ is a closed testing procedure}      \label{appVtildectp}
Consider the equivalence testing setting from Section \ref{secet}.
We will now show that $\tilde{V}$ coincides with a closed testing procedure. We assume the reader is familiar with the notation and theory from Section \ref{secbasicsct}, e.g.,  $\C$ is the collection of all nonempty subsets of $\{1,...,m\}$.
For every $I\in \C$, define 
$$
\R_I(t) := \{j\in I: |T_j|<\delta_j -t  \}, \quad  R_I(t) := |\R_I(t)|,
$$
$$
\R_I^-(t):=   \{j \in I: |T_j|>\delta_j + t  \}, \quad  R_I^-(t) = |\R_I^-(t)|.
$$

For every $I\in \C$,  we define the local test 
\begin{equation} \label{deflocalbasicet}
\lt'_I:= \mathbbm{1}( R_I >  R_I^-).
\end{equation}
It can easily be checked that if $I\subseteq \N$, then $\mathbb{P}(R_I(t)>R_I^-(t))\leq 0.5$.

 Recall the notation $\X$ from Section \ref{secbasicsct} .
 We now investigate which $I$ is in $\X$ for the closed testing procedure based on the local tests $\lt'_I$.

\begin{proposition}  \label{whichrejectet}
Consider the hypotheses $H_j: |\mu_j|\geq\delta_j$, $1\leq j \leq m$.
Consider the closed testing procedure based on the local tests defined by equation \eqref{deflocalbasicet}.
For any $I\in \C$, we have  $I\in\X$ if and only if  $R_I(t)>R^-(t)$.
\end{proposition}

Next we investigate what bounds the above closed testing procedure provides. For $I\in \C$, we have   
\begin{equation} \label{boundnotadmet}
t_{\alpha}(I)=  \max\{J \subseteq I:   J\not\in \X \} = \max\{J \subseteq I:   R_J\leq R^- \}.
\end{equation}
In particular, if we take $I=\R(t)=\R$, we obtain
$$t_{\alpha}(\R) = \max\{J \subseteq \R:   R_J\leq R^- \} =   \max\{J \subseteq \R:   |J|\leq R^- \} = R\wedge R^- =\tilde{V}(t).$$
Thus, $\tilde{V}(t)$ concides with the $50\%$-confidence bound for $|\N\cap\R|$ produced by this closed testing procedure.
Now consider any other $I\in \C$.
The quantity \eqref{boundnotadmet} equals 
$$\max\{J \subseteq I:   R_J\leq R^- \}  = |I|\wedge (   |\{j\in I: |T_j|\geq \delta_j-t  \}| + R^-   ) =  |I| \wedge (  |I\cap \R^c| +R^-    ),$$
where $\R^c=\{1,...,m\}\setminus \R$. Thus, $t_{\alpha}(I)$  is simply the bound  $R^- $ for $|\N\cap \R|$ plus the number of elements in $I$ that are not in $\R$. 
Hence, just as in Section \ref{secvtisctp}, we conclude that 
all the information contained in the bounds $t_{\alpha}(I)$, $I\in \C$, is already contained in our single bound $\tilde{V}$.  Thus, $\tilde{V}$ is equivalent to a closed testing procedure. 

This closed testing procedure is not admissible, because the underlying local tests $\lt'_I$ are not admissible, due to the possibility that $R_I = R_I^-$. However, we can again introduce a random coin flip $d$ in case $R_I =  R_I^-$, as in definition \eqref{defadmlt}. In case $R_I =  R_I^-$ and $d=1$,  the corresponding bound, $\tilde{V}'$ say,  will be $-\infty$ as in Section \ref{secnovel}. 
We conjecture that these redefined local tests are admissible. It then follows as in Section \ref{secvtildepadm}
that $\tilde{V}'$ is admissible.
Further, under additional assumptions, a slightly more powerful procedure can be defined.
 This is discussed in Appendix \ref{etadm}.

\section{An alternative FDP estimator for equivalence testing} \label{etadm}
This section is an extension of the theory of Section \ref{secet}.  Here, we introduce the following  assumption, under which we can make the method from Section \ref{secet} somewhat more powerful.  $\tilde{V}(t)$ will be defined as in Section \ref{secet}, except that the definition of the underlying quantity $R^-(t)$ will be slightly redefined.

\begin{assumption} \label{assmono}
Assumption \ref{asssym} holds and the test statistics $T_j$ with $j\in \N$ are mutually independent and each have a density function $f_j:\mathbb{R}\rightarrow [0,\infty)$ with respect to the Lebesgue measure. Moreover, $f_j$ is non-decreasing on $(\infty,\mu_j)$ and non-increasing on $(\mu_j,\infty)$.
\end{assumption}
Here, the assumption that $T_j$ with $j\in \N$ are mutually independent is more than  we need in the proof of Theorem \ref{thmetindep}. Future research may substantially weaken this independence assumption.

We still consider the setting of equivalence testing, so the null hypotheses are $H_j: \mu_j\geq \delta_j$, $1\leq j \leq m$.
Moreover, we still write $\R(t) = \{1\leq j \leq m: |T_j|<\delta_j -t  \}$ and $ R(t):= |\R(t)|.$
However, we now slighty modify the definition of  $R^-(t)$. We still use the same mathematical symbols as in Section \ref{secet},  to avoid an excess of notation.
Define
$$
\R^-(t):=   \{1\leq j \leq m: \delta_j + t <|T_j| \leq 3\delta_j-t \}, \quad  R^-(t):= |\R^-(t)|  
$$
and $\tilde{V} :=   R^- \wedge R$.

We now show that the bound  $\tilde{V}(t)$ is valid under Assumption \ref{assmono}.

\begin{theorem} \label{thmetindep}
Under Assumption \ref{assmono}, we have $\pr(V(t)\leq \tilde{V}(t))\geq 0.5$.
\end{theorem}

\begin{proof}
Let $\N^+= \{j\in \N: \mu_j\geq \delta_j\}$,  $\N^+= \{j\in \N: \mu_j\leq -\delta_j\}$.
Write $S_j^+ = (-\delta_j +t,\delta_j -t)$, $S_j^- =   [-3\delta_j+t ,  -\delta_j -t )    \cup (\delta_j + t ,  3\delta_j-t]$.

We have   $$ V(t)=| \{j\in \N: T_j\in S^+  \}  |,$$ and 
    $$ R_{\N}^-(t):=\N\cap R^-(t) =|\{j\in \N: T_j \in S_j^- \}  | .$$

Let $j\in \N$. Hence, either $\mu_j\geq \delta_j$ or $\mu_j\leq -\delta_j$. First suppose that $\mu_j\geq \delta_j$. 
Note that $\delta_j$ lies exactly in between the intervals $S_j^+$ and $(\delta_j + t ,  3\delta_j-t] \subset S_j^-$. Further, both those interval have width $2\delta_j-2t$.
Consequently, if $\mu_j=\delta_j$, then 
$$\pr(T_j\in S^+) = \int_{S^+_j} f_j(x) dx = \int_{(\delta_j + t ,  3\delta_j-t]} f_j(x) dx \leq \pr(T_j\in S^-),$$
since by Assumption \ref{assmono}, $T_j$ is unimodal and symmetric about $\mu_j=\delta_j$.
Now note that if $\mu_j>\delta_j$, then $f_j$ is unimodal and symmetric with mean larger than $\delta_j$, so that
$$\pr(T_j\in S^+) = \int_{S^+_j} f_j(x) dx \leq  \int_{(\delta_j + t ,  3\delta_j-t]} f_j(x) dx \leq \pr(T_j\in S^-).$$
We conclude that if $\mu_j\geq \delta_j,$ then  $\pr(T_j\in S^+) \leq \pr(T_j\in S^-)$. 
Analogously, we find that $\mu_j\leq -\delta_j,$ then  $\pr(T_j\in S^+) \leq \pr(T_j\in S^-)$. 

Since the $T_j$ with $j \in \N$ are mutually independent, it follows that 
$\pr(V(t)\leq R_{\N}^-(t)) \geq 0.5$. Hence, $\pr(V(t)\leq R^-(t)) \geq 0.5$, so that  $\pr(V(t)\leq \tilde{V}(t)) \geq 0.5$.
\end{proof}

We end this section by showing that $\tilde{V}$ coincides with a closed testing procedure.  Firstly, for every $I\in \C$, define $\R_I(t)$ as in Section \ref{secet} and define
$$
\R_I^-(t):=   \{j \in I: \delta_j + t < |T_j| \leq 3\delta_j-t  \}, \quad  R_I^-(t) = |\R_I^-(t)|.
$$
Like before, for every $I\in \C$ we now  define
$$
    \lt_I= \mathbbm{1}\big\{R_I(t) >  R_I^-(t)\big\}.
$$
Then the local tests $\lt_I $ are valid at level $\alpha=0.5$ in the sense that if $H_I$ is true, then the type I error rate is at most $0.5$. This can be proved in the same way as we proved that $\pr(V(t)\leq R_{\N}^-(t)) \geq 0.5$ in the proof of Theorem \ref{thmetindep}.

 Similar to Proposition \ref{propctp2}, we now investigate which $I$ is in $\X$ for the closed testing procedure based on the local tests $\lt_I$.
Letting $S^+$ and $S^-$ be as in the proof of Theorem \ref{thmetindep}, note that 
$$I\in\X \Leftrightarrow$$
$$  R_J> R^-_J\text{ for all } I\subseteq  J\in \C \Leftrightarrow $$
$$   \big|\big\{j\in J: T_j\in S^+   \big\}\big|  \geq   \big|\big\{j\in J: T_j\in S^- \big\}\big|              \text{ for all } I\subseteq  J\in\C \Leftrightarrow $$
$$   \big|\big\{j\in I: T_j\in S^+\big\}\big|  \geq   \big|\big\{j\in \{1,...,m\}: T_j\in S^- \big\}\big| \Leftrightarrow$$
$$R_I \geq R^-.$$
As in Section  \ref{secvtisctp} we find that for the current closed testing procedure 
$$t_{\alpha}(\R(t)) = \tilde{V}.$$
Thus, $\tilde{V}(t)$ concides with the $50\%$-confidence bound from the closed testing procedure based on local tests $\lt_I$.
Moreover, for any other  $I\in \C$,
$$t_{\alpha}(I) = \max\{|J|: J \subseteq I \text{ and }   R_J\leq \tilde{V}  \}  = $$
$$ |I|\wedge (  | \{j\in I: |T_j| < \delta_j -t  \}| + \tilde{V}'   ) =  |I| \wedge (  |I\cap \R^c| + \tilde{V}'    ).$$
 Thus, as in Section \ref{secvtildepadm}, the  $50\%$-confidence bound for the number of true hypotheses in $I$, $|\N\cap I|$, is simply the bound  $\tilde{V} $ plus the number of elements in $I$ that are not in $\R$. Thus,  $\tilde{V}$ is equivalent to the closed testing procedure based on local tests $\lt_I$.

\section{Proofs of results} \label{appproofs}

\subsection{Proof of Theorem \ref{thmflip}}

\begin{proof}
The statements about the case $\theta=0$ follow directly from the general group invariance testing principle, see e.g. \citt{hoeffding1952large}, \citt[][Theorem 17.2.1]{lehmann2022testing} and  \citt{hemerik2018exact}. 
Now suppose  $\theta<0$.  Let $\theta_n$ denote the $n$-vector $(\theta,...,\theta)'$.
Define $X^* = X -  \theta_n$. The entries of $X^*$ are symmetric with mean $0$ and hence 
\begin{equation} \label{eq:rewr}
\mathbb{P}(T_{id}(X^*)>T^{(1-\alpha) }(X^*))\leq \alpha
\end{equation} 
 by e.g. \citt{hemerik2018exact}.
Now note that 
\begin{equation} \label{diffid}
T_{id}(X^*)    = T_{id}(X) - n^{-1/2} n\theta,
\end{equation}
where $-n^{-1/2} n\theta$ is positive.
Further, for any transformation $g\in \G$ we have 
$$   T_{g}(X^*)  \leq  T_{g}(X) - n^{-1/2} n\theta$$
The latter implies that
\begin{equation} \label{diffq}
T^{(1-\alpha) }(X^*)  \leq T^{(1-\alpha) }(X) -n^{-1/2} n\theta.  
\end{equation}
From expressions  \eqref{diffid} and  \eqref{diffq} and the fact that   $- n^{-1/2} n\theta>0$,   it follows that the following implication always holds:
$$  T_{id}(X)>T^{(1-\alpha) }(X)  \implies T_{id}(X^*)>T^{(1-\alpha) }(X^*).$$
Hence, $$\pr\Big\{\text{reject } H_0\Big\}= \pr\Big\{T_{id}(X)>T^{(1-\alpha) }(X)\Big\}  \leq \pr \Big\{T_{id}(X^*)>T^{(1-\alpha) }(X^*)\Big\} \leq \alpha.$$
by inequality \eqref{eq:rewr}. This finishes the proof.
\end{proof}

\subsection{Proof of Theorem  \ref{thm2gr}}

\begin{proof}
In the literature it is quite common to consider the point null hypothesis $H_0: \theta = 0$. As remarked, in that case the results directly follow from e.g. \citt{hemerik2018exact}. Now suppose  $\theta < 0$. We  define $X^*=\{(Z- \theta^1_n)' ,(Y-\theta^2_n)'\}'$, where $\theta^1_n$ and $\theta^2_n$ denote the $n$-vectors $(\theta^1,...,\theta^1)'$ and $(\theta^2,...,\theta^2)'$ respectively. Then the distribution of $X^*$ is permutation-invariant, so that by e.g. \citt{hemerik2018exact},
$$
\mathbb{P}_{H_0}(T_{id}(X^*)>T^{(1-\alpha) }(X^*))\leq \alpha.
$$
Now note that 
\begin{equation} \label{diffid2}
T_{id}(X^*)   =T_{id}(X) - n^{-1/2}( n \theta^1 - n \theta^2),
\end{equation}
where $- n^{-1/2}( n \theta^1 - n \theta^2)>0$ , since $\theta^1-\theta^2=\theta<0$.
Further, for any transformation $g\in \G$ we have 
$$T_{g}(X^*)  \leq T_{g}(X) - n^{-1/2}( n \theta^1 - n \theta^2).$$
The latter implies that
\begin{equation} \label{diffq2}
T^{(1-\alpha) }(X^*)   \leq T^{(1-\alpha) }(X) - n^{-1/2}( n \theta^1 - n \theta^2).
\end{equation}
From results  \eqref{diffid2} and  \eqref{diffq2} it follows that the following implication always holds:
$$  T_{id}(X)>T^{(1-\alpha) }(X)  \implies T_{id}(X^*)>T^{(1-\alpha) }(X^*).$$
The result now follows as at the end of the proof of Theorem \ref{thmflip}.

\end{proof}

\subsection{Proof of Theorem \ref{thmnewvalid}}

 \begin{proof}

 We have 
 \begin{align*}
  V(t) =&   |\{ j\in \N: T_j  > \delta_j+t \}  | \\
 \leq &  |\{ j\in \N: T_j  > \mu_j+t \}  | \\
  = & | \{j\in \N: T_j -\mu_j >  t \}  |=:N_1, 
  \end{align*}
  where we used that if $j\in \N$, then $\delta_j  \geq \mu_j$.

   Note that  
    \begin{align*}
    R^-(t) =& |\{1\leq j\leq m: T_j  < \delta_j -t\}  |\\
    \geq & |\{j\in \N: T_j  < \delta_j-t\}  | \\
    \geq & |\{j\in \N: T_j  -\mu_j  <   -t\}  | =:N_2,
    \end{align*}
    where we again used that if $j\in \N$ then $\delta_j  \geq \mu_j$.
 
   Thus, $V(t)\leq N_1$ and $R^- \geq N_2$.
   Consequently,  
   \begin{align*}
   &\mathbb{P}\big\{V(t)\leq   R^-(t)\big\} \geq\\
   &  \mathbb{P}\big\{N_1\leq N_2\big\} =\\
   &  \mathbb{ P} \Big\{      | j\in \N: T_j -\mu_j >  t \}  |  \leq     |\{j\in \N: T_j  -\mu_j <   -t\}  |     \Big\}. 
   \end{align*}   
  By Assumption  \ref{asssym}, the above equals
  $$  \mathbb{ P} \Big\{      | j\in \N: -(T_j -\mu_j) >  t \}  |  \leq     |\{j\in \N: -(T_j  -\mu_j) <   -t\}  |     \Big\} =$$
  $$  \mathbb{ P} \Big\{      | i\in \N: T_j -\mu_j < - t \}  |  \leq     |\{i\in \N: T_j  -\mu_j  >   t\}  |     \Big\} $$
  but this is exactly $\mathbb{P}\{N_2\leq N_1\}$. Thus  $\mathbb{P}(N_1\leq N_2) = \mathbb{P}(N_2\leq N_1)  $. Since the sum of these probabilities is at least 1, It follows that $\mathbb{P}(N_1\leq N_2)\geq 0.5$. 
  
  Note that since $V(t)\leq R(t)$, we have  $$V(t)\leq   \tilde{V}(t) \quad \Leftrightarrow  \quad V(t)\leq  R^-(t)   \wedge  R(t) \quad \Leftrightarrow \quad V(t)\leq   R^-(t).$$ 
Thus, $$\mathbb{P}\big\{V(t)\leq   \tilde{V}(t)\big\} = \mathbb{P}\big\{V(t)\leq   R^-(t)\big\} \geq  \mathbb{P}\big\{N_1\leq  N_2\big\}  \geq 0.5.$$ 
Note that $V(t)\leq   \tilde{V}(t)$ implies $FDP(t)\leq   \tilde{FDP}(t)$, so that we also have 
$$\mathbb{P}(FDP(t)\leq   \tilde{FDP}(t))\geq 0.5.$$
 \end{proof}

\subsection{Proof of Lemma  \ref{lemmaties}}

\begin{proof}
We have
\begin{align*}
V(t) =& |\{ j\in \N: T_j  > \delta_j+t \}  | \\
=& | \{j\in \N: T_j -\mu_j >  t \}  |
\end{align*}
  and 
  \begin{align*}
R^-(t) =& |\{1\leq j\leq m: T_j  < \delta_j -t\}  | \\
=&|\{j\in \N: T_j  < \delta_j-t\}  | \\
=& |\{j\in \N: T_j  -\mu_j  <   -t\}  |.
\end{align*}
    Thus,

$$ \mathbb{P}\{V(t)\leq   R^-(t)\} =$$
 $$ \mathbb{P}\Big\{| \{j\in \N: T_j -\mu_j >  t \}  |\leq    |\{j\in \N: T_j  -\mu_j  <   -t\}  | \Big\}.$$  
We now reason analogously to the proof of Theorem \ref{thmnewvalid}.   By Assumption  \ref{asssym}, the above equals
 $$ \mathbb{P}\big\{| \{j\in \N: -(T_j -\mu_j) >  t \}  |\leq    |\{j\in \N: -(T_j  -\mu_j)  <   -t\}  | \big\} =$$  
$$  \mathbb{ P} \Big\{      | i\in \N: T_j -\mu_j < - t \}  |  \leq     |\{i\in \N: T_j  -\mu_j  >   t\}  |     \Big\} $$
but this is exactly $\mathbb{P}\{ R^-(t) \leq V(t)\}$.
Hence, $\mathbb{P}\{V(t)\leq   R^-(t)\} =\mathbb{P}\{V(t)\leq   R^-(t)\} $, which means that these probabilities are at least $0.5$. Note that  $\mathbb{P}(E_1)+ \mathbb{P}(E_2) + \mathbb{P}(E_3) =1$ and $\mathbb{P}(E_1) = \mathbb{P}(E_3).$ 
Thus, $2\mathbb{P}(E_1)  + \mathbb{P}(E_2)=1$, so that $\mathbb{P}(E_1) =   0.5  - 0.5\mathbb{P}(E_2) )$.
\end{proof}

\subsection{Proof of Proposition  \ref{basicsvbarprime}}

\begin{proof}

\emph{Proof of first statement.}
Define the random variables $N_1$ and $N_2$ as in the proof of Theorem \ref{thmnewvalid}:

\begin{align*}
N_1=&| \{j\in \N: T_j -\mu_j >  t \}  |,\\
N_2=&|\{j\in \N: T_j  -\mu_j  <   -t\}  | .
\end{align*}

We have 
\begin{align*}
&\mathbb{P}\{V(t) \leq   \tilde{V}'(t) |   R(t) \neq R^-(t)\}  = \\
&\mathbb{P}\{V(t) \leq   \tilde{V}(t) |   R(t) \neq R^-(t)\}  = \\
&\mathbb{P}\{V(t) \leq   R^-(t) |   R(t) \neq R^-(t)\}.
\end{align*}

As in the proof of Theorem \ref{thmnewvalid}, we have $V(t)\leq N_1$ and $R^-(t)\geq N_2$.
 Hence the above is at least 
 \begin{align*}
& \mathbb{P}\{N_1 \leq   N_2 |   R(t) \neq R^-(t)\}  = \\
&\mathbb{P}\{N_2 \leq   N_1|   R(t) \neq R^-(t)\}  \geq \\
&\mathbb{P}\{R^-(t) \leq  V(t)|   R(t) \neq R^-(t)\}. 
\end{align*}
  Since 
  $$\mathbb{P}\Big\{V(t) \leq   R^-(t) \big|   R(t) \neq R^-(t)\Big\}  +  \mathbb{P}\Big\{R^-(t) \leq  V(t)\big|   R(t) \neq R^-(t)\Big\} \geq 1,$$
it follows that 
$$\mathbb{P}\big\{V(t) \leq   R^-(t) |   R(t) \neq R^-(t)\big\} \geq 0.5$$ 
and hence
$$\mathbb{P}\big\{V(t) \leq   \tilde{V}'(t) |   R(t) \neq R^-(t)\big\}  \geq 0.5.$$
Further, 
$$
\mathbb{P}\big\{V(t) \leq   \tilde{V}'(t) |   R(t) = R^-(t)\big\}    = \mathbb{P}\big\{d=1 |   R(t) = R^-(t)\big\}=0.5.
$$
From the law of total probability it follows that $\mathbb{P}\{V(t) \leq   \tilde{V}'(t) \} \geq 0.5.$
\\
\\
\emph{Proof of second statement.}
Suppose that for all $1\leq j \leq m$ we have $\mu_j=\delta$. Then $V(t)=R(t)$.
Define $E_1$, $E_2$ and $E_3$ as in Lemma \ref{lemmaties}.
Note that $$\mathbb{P}\{V(t)  >  \tilde{V}'(t)\} =$$
$$    \mathbb{P}\Big\{V(t)  >  \tilde{V}'(t) \text{ and } V(t) \neq  R^-(t)   \Big \}  +     \mathbb{P}\Big\{V(t)  >  \tilde{V}'(t) \text{ and } V(t) =  R^-(t)    \Big\}.$$
Since $R(t)=V(t)$, the above equals 
 \begin{align*}
&\mathbb{P}\Big\{V(t)  >  \tilde{V}'(t) \text{ and } R(t) \neq  R^-(t)   \Big \}   +     \mathbb{P}\Big\{V(t)  >  \tilde{V}'(t) \text{ and } R(t) =  R^-(t)    \Big\}= \\
 & \mathbb{P}\Big\{V(t)  >  \tilde{V}(t) \text{ and } R(t) \neq  R^-(t)   \Big \}   +     \mathbb{P}\Big\{ R(t) =  R^-(t)  \text{ and }   d=-\infty  \Big\}= \\
 &\mathbb{P}\Big\{V(t)  >  R^-(t) \text{ and } R(t) \neq  R^-(t)   \Big \}   +     0.5 \mathbb{P}\{R(t) =  R^-(t)\} =
\end{align*}
\begin{equation} \label{eqrmin}
 \mathbb{P}\{V(t) >  R^-(t)\} + 0.5 \mathbb{P}\{V(t) =  R^-(t)\},
\end{equation}
where we used that $V(t)=R(t)$.
Recall that $\{V(t) >  R^-(t)\}=E_3$ and   $\mathbb{P}\{V(t) =  R^-(t)\}=E_2.$

Now, by   Lemma \ref{lemmaties} we have $\mathbb{P}(E_1) = \mathbb{P}(E_3)$, so that $\mathbb{P}(E_3) = 0.5\mathbb{P}(E_1) +0.5\mathbb{P}(E_3)$.
Thus,   \eqref{eqrmin} equals 
 \begin{align*}
& \mathbb{P}(E_3)+0.5\mathbb{P}(E_2) = \\
& 0.5\mathbb{P}(E_1) +0.5\mathbb{P}(E_3) +0.5\mathbb{P}(E_2)= \\
&0.5(\mathbb{P}(E_1) +\mathbb{P}(E_2) +\mathbb{P}(E_3) )=0.5\cdot 1=0.5,
\end{align*}
 as was to be shown.

\end{proof}

\subsection{Proof of Proposition \ref{ltval}}

\begin{proof}
For any true intersection hypothesis $H_I$ with $I\in\C$,
$$\mathbb{P}(\lt_I=1) = \mathbb{P}(R_I>R^-_I)=$$
$$  \mathbb{P}\Big(|\{j\in I: T_j >\delta_j +t   \}|>   |\{j\in I: T_j <\delta_j -t   \}|\Big).$$
Note that for every $j\in I$,  $\mu_j\leq \delta_j$, since $H_i$ is true. Hence the above is at most
$$  \mathbb{P}\Big(|\{j\in I: T_j -\mu_j > t   \}|>   |\{j\in I: T_j -\mu_j< -t   \}|\Big).$$
Similarly to other proofs, by symmetry of $(T_j-\mu_j:j\in I)$ the above is at most $0.5=\alpha$.
This finishes the proof.
\end{proof}

\subsection{Proof of Proposition  \ref{whichinX}}

\begin{proof}
For $I\in \C$ we have 
 \begin{align*}
 I\in\X \Leftrightarrow &   \lt_J=1 \text{ for all } I\subseteq  J\in\C  \\
 \Leftrightarrow  & R_J>R^-_J\text{ for all } I\subseteq  J\in\C  \\
\Leftrightarrow  & |\{j\in J: T_j>\delta_j+t\}|  >   |\{j\in J: T_j<\delta_j-t\}|              \text{ for all } I\subseteq  J\in \C  \\
\Leftrightarrow  & |\{i\in I: T_j>\delta_j+t\}|  >   |\{j\in \{1,...,m\}: T_j<\delta_j-t\}|            \\
\Leftrightarrow  & R_I > R^-.
\end{align*}
Thus, $I\in\X$ if and only if  $R_I>R^-$.
\end{proof}

\subsection{Proof of Lemma  \ref{exhaust}}

\begin{proof}
Let $I\in\C$. We must give an example of a distribution of  $(T_1,...,T_m)$  satisfying Assumption \ref{asssym} for which $\mathbb{P}(\delta'_I=1)=0.5$. In fact, it holds for any distribution  that satisfies Assumption \ref{asssym} and in addition satisfies $\mu_j=\delta_j$ for all $1\leq j \leq m$.
Indeed, then   $(T_j-\delta_j: j\in I) \,{\buildrel d \over =}\,    -(T_j-\delta_j: j\in I) $ and hence $\pr(R_I(t) >  R_I^-(t))=\pr(R_I(t) <  R_I^-(t)) $. Further, $\pr(b=1)=0.5$. Consequently,  $\mathbb{P}(\lt'_I=1)=0.5$.
 \end{proof}

\subsection{Proof of  Lemma  \ref{ltIadm}}

\begin{proof}
Suppose such a $\phi^*_I$ exists and pick one. Since  $\mathbb{P}( \phi^*_I >\lt'_I )>0$, we can choose a subset $A\subset \mathbb{R}^n$ with strictly positive Lebesgue measure such that if $(t_1,...,t_m)\in A$, then  $\phi^*_I((t_1,...,t_m))=1$ and   $\lt'_I((t_1,...,t_m))=0$.
Since the support of $(T_1,...,T_m)$ is generally $\mathbb{R}^m$, it follows that generally 
 $\mathbb{P}(\lt'_I=1)<  \mathbb{P}(\phi^*_I=1)$, so that  if $H_I$ is true, then  $\mathbb{P}(\lt'_I=1)<0.5$ . However, by Lemma  \ref{exhaust} this cannot generally hold. Thus, there is no such test $\phi^*_I$, as was to be shown.
 \end{proof}

\subsection{Proof of Proposition  \ref{propctp2}}

\begin{proof}
Let $I\in \C$.
Note that for $J\in \C$, $\lt_J'=1$ means that either $ R_J>R^-_J$ or $\{R_J=R^-_J\}\cap\{b=1\}$ holds. 
Suppose that $b=1$. Then  $\lt_J=1$ if and only if $R_J\geq R^-_J$. Hence, then 
$$I\in\X \Leftrightarrow R_I \geq R^-$$
exactly as in the proof of Proposition \ref{whichinX}.
Now suppose that $b=0$. Then it analogously follows that $I\in\X \Leftrightarrow R_I > R^-$.

We see that in general $ \mathbbm{1}(I\in\X) = \mathbbm{1}(R_I > R^-) + b  \mathbbm{1}(R_I = R^-)$. Consider the following events, whose union has probability 1.

\begin{itemize}
\item \emph{Event 1:} $R_I=R_I^-$ and $b=1$. Then $\tilde{V}' =d=-\infty< R_I$.

\item \emph{Event 2:} $R_I > R^-$. Then $ R^-< R_I  \leq R$ and hence  $\tilde{V}'=R \wedge R^-= R^- <R_I$.
\item \emph{Event 3:} $R_I < R^-$  and $\{R_I\neq R_I^-\}\cup\{b=0\}$ and $R\leq R^-$. Then $\tilde{V}'= R \wedge R^-=R\geq R_I$.

\item \emph{Event 4:} $R_I < R^-$  and $\{R_I\neq R_I^-\}\cup\{b=0\}$ and $R>  R^-$. Then $\tilde{V}'= R \wedge R^-=R^->R_I.$

\item  \emph{Event 5:} $R_I = R^-$ and  $\{R_I\neq R_I^-\}\cup\{b=0\}$. Then $\tilde{V}' = \tilde{V} =R \wedge R^- =R^-=R_I$.

\end{itemize}
Consequently, $R_I > \tilde{V'}$ if and only if event 1 or 2 happens. 
 We conclude that 
$$\mathbbm{1}(I\in\X) =   \mathbbm{1}(R_I > R^-) + b  \mathbbm{1}(R_I = R^-) =\mathbbm{1}(R_I > \tilde{V}'),$$ 
so that we are done.

\end{proof}

\subsection{Proof of Theorem  \ref{thmadm}}

\begin{proof}
As already shown in the main text, the closed testing procedure based on the local tests $\lt'_I$, $I\in \C$, is admissible in the sense of  \citt{goeman2021only}.
Now suppose  $\tilde{V}'$ is not admissible. Choose a procedure  producing bounds $\tilde{V}^*(t)$ with the three mentioned properties stated in the theorem. Further, pick a distribution of $(T_1,...,T_m)$  satisfying Assumption \ref{asssym} such that  $\mathbb{P}\{ \tilde{V}'(t) >\tilde{V}^*(t)\}>0$. We discussed in the main text that $\tilde{V}'(t)$ is equivalent to an admissible closed testing procedure, namely the procedure that provides the upper bound $|I| \wedge (  |I\cap \R^c| + \tilde{V}'    )$ for every $I\in \C$. Note that $\tilde{V}^*(t)$ is equivalent to a procedure  that is uniformly better, namely the procedure that provides the upper bound $|I| \wedge (  |I\cap \R^c| + \tilde{V}^*    )$ for  every $I\in \C$. This is a contradiction with the admissibility of the former procedure. This finishes the proof.
\end{proof}

\subsection{Proof of Theorem  \ref{boundetvalid}}
 
 \begin{proof}
 Write 
 \begin{align*}
  \N^- =&\{j\in \N: \mu_j\leq -\delta_j\},\\
 \N^+ =&\{j\in \N: \mu_j\geq\delta_j\}.
 \end{align*}
 Note that $\N= \N^-\cup\N^+$.
 
 We have $V(t) = $
 \begin{align*}
 | j\in \N: |T_j|  & < \de_j-t \}  | =\\
 \Big|\{  j\in \N^-: |T_j|  < \de_j-t\}  & \cup  \{  j\in \N^+: |T_j|  < \de_j-t\}  \Big| \leq\\
 \Big|\{  j\in \N^-: T_j  > -(\de_j-t)\}  & \cup  \{  j\in \N^+: T_j  < \de_j-t\}  \Big|  \leq\\
 \Big |\{  j\in \N^-: T_j -  \mu_j >  t\} &  \cup  \{  j\in \N^+: T_j -  \mu_j < - t\}\Big|=:Z_1,  
 \end{align*}
where the last step follows from the fact that  if $j \in \N^-$, then $\mu_j \leq -\delta_j$ so $-\mu_j \geq \delta_j$, and if  $ j\in \N^+$, then $\mu_j\geq \delta_j$ so $-\mu_j\leq -\delta_j$.
 
 We have  $\tilde{V}(t) =$
  \begin{align*}
|\{1\leq j \leq m: |T_j|  & > \de_j+t\}  | \geq\\
 |\{j\in \N: |T_j|  & \geq \de_j+t\}  |  =\\
\Big|   \{i\in \N: T_j < -\de_j-t\}  & \cup    \{j\in \N: T_j > \de_j+t\}       \Big|   \geq\\
\Big| \{i\in \N^-: T_j < -\de_j -t \}  & \cup   \{j\in \N^+: T_j >\de_j+t\}    \Big|  \geq\\
\Big|   \{i\in \N^-: T_j   -  \mu_j < -t \} & \cup  \{j\in \N^+: T_j -  \mu_j  > t \}      \Big|  =:Z_2  ,
 \end{align*}
    where the last inequality follows from the fact that if $j\in \N^-$ then $\mu_j\leq -\de_j$ so $-\mu_j\geq \de_j$ and  if $j\in \N^+$ then $\mu_j\geq \de_j$ so $-\mu_j\leq -\delta_j$.

 We consequently have  
   \begin{align*}
&\mathbb{P}(V(t)\leq   \tilde{V}(t)) \geq  \mathbb{P}(Z_1\leq Z_2)=\\
&\mathbb{ P} \Big\{\big|\{  j\in \N^-: T_j-  \mu_j >  t\}  \cup  \{  i\in \N^+: T_j -  \mu_j < - t\}\big| \leq \\
&\big|   \{i\in \N^-: T_j   -  \mu_j< -t \}  \cup   \{j\in \N^+: T_j -  \mu_j > t \}   \big| \Big\}.
 \end{align*}
By Assumption \ref{asssym}, the above equals
   \begin{align*}
&\mathbb{ P} \Big\{\big |\{  j\in \N^-: T_j -  \mu_j <-  t\}  \cup  \{  i\in \N^+: T_j -  \mu_j >  t\}\big| \leq\\
& \big| \{j\in \N^-: T_j   -  \mu_j > t \}  \cup   \{j\in \N^+: T_j -  \mu_j  <- t \} \big|   \Big\}, 
 \end{align*}
 but this is $\mathbb{P}(Z_2\leq Z_1)$. Thus  $\mathbb{P}(Z_1\leq Z_2) = \mathbb{P}(Z_2\leq Z_1)  $. Since the sum of these probabilities is at least 1, It follows that $\mathbb{P}(Z_1\leq Z_2)\geq 0.5$. Consequently, $\mathbb{P}(V(t)\leq   \tilde{V}(t)) \geq 0.5$, as was to be shown.
 \end{proof}

\subsection{Proof of Proposition  \ref{whichrejectet}}

\begin{proof}
For $I\in \C$ we have 
\begin{align*}
I\in\X \Leftrightarrow &  \lt'_J=1 \text{ for all } I\subseteq  J\in\C \\
\Leftrightarrow & R_J>R^-_J\text{ for all } I\subseteq  J\in \C \\
\Leftrightarrow & |\{j\in J: |T_j|<\delta_j-t\}|  >   |\{j\in J: |T_j|>\delta_j + t\}|         \text{ for all } I\subseteq  J\in \C. 
\end{align*}
Note that the above holds if and only if 
$R_I > R^- $.
 Indeed, to show the first implication, suppose the above holds. Let $K = \{1\leq j\leq m: |T_j|>\delta_j + t\}$.
Then 
$$R_I = R_{I\cup K} > R_{I\cup K}^- = R_{\{1,...,m\}}^-=R^-.$$
To show the second implication, note that  if $R_I > R^-$, then for all $I\subseteq  J\in \C$, we have  $$R_J\geq R_I > R^-\geq R^-_J.$$
We conclude that $I\in\X$ if and only if  $R_I>R^-$.
\end{proof}

\subsection{Proof of Theorem  \ref{thmcontrol}}

\begin{proof}
We first provide a proof for the case of one-sided testing  (Section \ref{secnovel}). We then give the proof for the case of equivalence testing (Section \ref{secet}), which is analogous.
\\
\\
\noindent \textbf{Proof for the case of one-sided testing}\\
We first define some quantities. Then we give an overview of the three steps of the proof. Finally, we prove the three steps in detail.
For $t\geq 0$, let
\begin{align*}
&A^-(t) =|\{j\in\N^c: T_j<\delta_j-t\}|,\\
&A^+(t) =|\{j\in\N^c: T_j>\delta_j+t\}|,\\
&N^-(t) =|\{j\in\N: T_j<\delta_j-t\}|,\\
&N^+(t) =|\{j\in\N: T_j>\delta_j+t\}|,
\end{align*}
where in fact $N^+=V$. Further, note that $R= N^+ + A^+$ and $R^- =N^- + A^- $.
Observe  that
$$FDP=\frac{N^+}{(N^+ + A^+)\vee 1}.$$
Define 
$$\tilde{FDP}^* := \frac{N^- }{(N^+ + A^+) \vee 1}\wedge 1  \leq  \frac{N^- + A^-}{(N^+ + A^+)\vee 1}\wedge 1 =\tilde{FDP}.$$
Note that 
$$FDP> \gamma \Leftrightarrow N^+ > \gamma({N^+ + A^+}) \Leftrightarrow $$
$$ N^+ -\gamma N^+ > \gamma A^+ .$$  
Likewise,
$$  \tilde{FDP}^*  > \gamma   \Leftrightarrow            N^- > \gamma(N^+ + A^+)    \Leftrightarrow$$
$$   N^- -\gamma N^+ > \gamma A^+.$$
Let $$t_1 = \sup\big\{t\geq 0:  FDP(t)>\gamma\} = \sup\{t\geq 0: N^+(t ) - \gamma N^+(t )  >   \gamma A^+(t) \big\},$$
 $$t_2 = \sup\big\{t\geq 0:  \tilde{FDP}^*(t)>\gamma\} = \sup\{t\geq 0:   N^-(t) -\gamma N^+(t) > \gamma A^+(t)  \big\},$$
 $$t_3 = \sup\big\{t\geq 0: N^-(t ) - \gamma N^-(t )  >   \gamma A^+(t) \big\}.$$
\\
\\
\emph{Overview of the proof.}
In Step 1 of this proof, we show that  $\pr(t_1\leq t_3)\geq 0.5$. 
In Step 2, we show that  $\pr(t_1\leq t_2|t_1\leq t_3 )=1$, so that $\pr(t_1\leq t_2)\geq 0.5$. In Step 3, we show that $\pr(\forall t\geq s^+: FDP(t)\leq \gamma|t_1\leq t_2) =1$.
It thus follows that $\pr(\forall t\geq s^+: FDP(t)\leq \gamma)\geq 0.5$, as we want.
\\
\\
\emph{Step 1.}
Note that 
$$\pr(t_1\leq t_3) = $$
$$\pr\Big[\sup\big\{t\geq 0: (1- \gamma) N^+(t )  >   \gamma A^+(t) \big\}  \leq \sup\big\{t\geq 0:(1- \gamma)N^-(t )  >   \gamma A^+(t) \big\}\Big] =$$
\begin{align*}
\pr\Big[&\sup\big\{t\geq 0: (1- \gamma) |\{j\in\N: T_j>\delta_j+t\}|  >   \gamma A^+(t) \big\}  \leq \\
&\sup\big\{t\geq 0:(1- \gamma)|\{j\in\N: T_j<\delta_j-t\}|  >   \gamma A^+(t) \big\}\Big]
\end{align*}

For all $j\in \N$, $\mu_j\leq \delta_j$. Hence, replacing $\delta_j$ by $\mu_j$ makes $|\{j\in\N: T_j>\delta_j+t\}|$  larger and $|\{j\in\N: T_j<\delta_j-t\}|$ smaller, so that the above probability becomes smaller. Thus, the above is larger than or equal to 
\begin{equation} \label{eqsymmandindep1}
\begin{aligned}
\pr\Big[&\sup\big\{t\geq 0: (1- \gamma) |\{j\in\N: T_j>\mu_j+t\}|  >   \gamma A^+(t) \big\}  \leq \\
&\sup\big\{t\geq 0:(1- \gamma)|\{j\in\N: T_j<\mu_j-t\}|  >   \gamma A^+(t) \big\}\Big]
\end{aligned}
\end{equation}
By Assumption \ref{asssym} and by independence of $(T_i:i\in \N)$ from $(T_i:i\in \N^c)$, the above equals
\begin{equation} \label{eqsymmandindep2}
\begin{aligned}
\pr\Big[&\sup\big\{t\geq 0: (1- \gamma) |\{j\in\N: -T_j>-\mu_j+t\}|  >   \gamma A^+(t) \big\}  \leq \\
&\sup\big\{t\geq 0:(1- \gamma)|\{j\in\N: -T_j<-\mu_j-t\}|  >   \gamma A^+(t) \big\}\Big] 
\end{aligned}
\end{equation}
$$= $$
\begin{align*}
\pr\Big[&\sup\big\{t\geq 0: (1- \gamma) |\{j\in\N: T_j<\mu_j-t\}|  >   \gamma A^+(t) \big\}  \leq \\
&\sup\big\{t\geq 0:(1- \gamma)|\{j\in\N: T_j>\mu_j+t\}|  >   \gamma A^+(t) \big\}\Big] 
\end{align*}
$$\geq $$
\begin{align*}
\pr\Big[&\sup\big\{t\geq 0: (1- \gamma) |\{j\in\N: T_j<\de_j-t\}|  >   \gamma A^+(t) \big\}  \leq \\
&\sup\big\{t\geq 0:(1- \gamma)|\{j\in\N: T_j>\de_j+t\}|  >   \gamma A^+(t) \big\}\Big] 
\end{align*}
$$=$$
$$ \pr\Big[\sup\big\{t\geq 0: (1- \gamma) N^-(t)  >   \gamma A^+(t) \big\}  \leq \sup\big\{t\geq 0:(1- \gamma)N^+(t) >   \gamma A^+(t) \big\}\Big]= $$
$$ \pr(t_3\leq t_1).$$
Thus,  $\pr(t_1\leq t_3) \geq   \pr(t_3\leq t_1)$. We conclude that  $\pr(t_1\leq t_3)\geq 0.5$.
\\
\\
\emph{Step 2.} 
Let $E$ be the event that $t_1\leq t_3$. Suppose $E$ happens. Note that either $t_3=-\infty$ or $t_3>0$. In the former case, we also have $t_1=-\infty$ and hence $t_1\leq t_2$.
Let $E'=E\cap\{t_3>0\}$.

Note that at $t_3$, the function $N^-(\cdot)$ has a jump downwards.
Since with probability one $|T_1-\delta_1|,...,|T_m-\delta_m|$ are distinct,  with probability 1, $N^+(\cdot)$ has no discontinuity at $t_3$.
Hence, $\pr(t_1<t_3|E')=1.$

Thus, under $E'$, with probability 1, 
\begin{equation} \label{eqinter1}
\lim_{\epsilon\downarrow 0}\Big(N^+(t_3-\epsilon ) - \gamma N^+(t_3-\epsilon )\Big) \leq \gamma A^+(t_3) < \lim_{\epsilon\downarrow 0}\Big(N^-(t_3-\epsilon ) - \gamma N^-(t_3-\epsilon )\Big).
  \end{equation}
But that means that under $E'$, with  probability 1, 
 $$\lim_{\epsilon\downarrow 0}N^+(t_3-\epsilon )< \lim_{\epsilon\downarrow 0}N^-(t_3-\epsilon ),$$
so that  
\begin{equation} \label{eqinter2}
 - \lim_{\epsilon\downarrow 0} \gamma N^+(t_3-\epsilon )\geq - \lim_{\epsilon\downarrow 0} \gamma N^-(t_3-\epsilon ).
   \end{equation}
Combining     \eqref{eqinter2} with the last inequality of \eqref{eqinter1}, we find that
 under $E'$, with  probability 1, 
\begin{equation} \label{eqinter3}
   \lim_{\epsilon\downarrow 0}(N^-(t_3-\epsilon ) - \gamma N^+(t_3-\epsilon )) \geq   \lim_{\epsilon\downarrow 0}(N^-(t_3-\epsilon ) - \gamma N^-(t_3-\epsilon )) >   \gamma A^+(t_3).
   \end{equation}

Under $E'$, with probability 1 there is a neighbourhood of $t_3$ where $A^+(\cdot)$ does not have a discontinuity.
Hence under $E'$ with probability 1 there is an $\eta>0$ such that for all $0<\epsilon<\eta$,
$$N^-(t_3-\epsilon ) - \gamma N^+(t_3-\epsilon )> \gamma A^+(t_3-\epsilon).$$
Hence, under $E'$ with probability 1,  $t_2\geq t_3$ and hence 
$t_1<t_3\leq t_2$. 
We conclude that  
$$\pr\big(t_1\leq  t_2\big)\geq  \pr(E')\pr\big(t_1\leq  t_2|E'\big) +  \pr\big(E\setminus E'\big)\pr\big(t_1\leq t_2|E\setminus E'\big) \geq$$
$$   \pr\big(E'\big)\cdot 1  + \pr\big(E\setminus E'\big) \geq  0.5,$$
so that Step 2 is completed.
\\
\\
\emph{Step 3.} 
Suppose $t_1\leq  t_2$. 
If $t_1=-\infty$, then $FDP(t)\leq \gamma$ for all $t\geq 0$ and hence $FDP(t)\leq \gamma$ for all $t\geq s^+$.
Now suppose $t_1\geq 0$. 
Then, with conditional probability 1, $t_1< t_2$, since with probability 1 the discontinuities of $N^+(\cdot)$ and $N^-(\cdot)$ do not overlap. 
Then, we have 
$$0\leq t_1=  \sup\{t\geq 0:  FDP(t)>\gamma\}   < t_2 = $$
$$ \sup\{t\geq 0:  \tilde{FDP}^*(t)>\gamma\} \leq  \sup\{t\geq 0:  \tilde{FDP}(t)>\gamma\} = s^+$$
by construction of $s^+$.

Thus, $$\sup\{t\geq 0:  FDP(t)>\gamma\} < s^+,  $$
so that $\forall t\geq s^+: FDP(t)\leq \gamma.$

Thus, $\pr(\forall t\geq s^+: FDP(t)\leq \gamma |  t_1\leq t_2)=1$, where $\pr(t_1\leq t_2)\geq 0.5$. We conclude that 
$\pr(\forall t\geq s^+: FDP(t)\leq \gamma)\geq 0.5$, as was to be shown.
\\
\\
\noindent \textbf{Proof for the case of equivalence testing}\\
The proof is essentially the same, except that for $t\geq 0$ we now define
\begin{align*}
&A^-(t) =|\{j\in\N^c: |T_j|>\delta_j+t\}|,\\
&A^+(t) =|\{j\in\N^c: |T_j|<\delta_j-t\}|,\\
&N^-(t) =|\{j\in\N: |T_j|>\delta_j+t\}|,\\
&N^+(t) =|\{j\in\N: |T_j|<\delta_j-t\}|
\end{align*}
Note that $A^+(t)$ and $N^+(t)$ are simply 0 if $t\geq\max\{\delta_j:1\leq j \leq m\}$.
Define $\tilde{FDP}^*(\cdot)$, $t_1$, $t_2$ and $t_3$ as before, but  with these new definitions of $A^-(t),...,N^+(t)$. 

As above, it follows from Assumption \ref{asssym} and independence of $(T_i:i\in \N)$ from $(T_i:i\in \N^c)$ that $\pr(t_1\leq t_3)\geq 0.5$, which finishes Step 1. Steps 2 and 3 are the same as above. This finishes the proof.

\end{proof}

\subsection{Proof of Theorem \ref{thmasymptcontrol}}

\begin{proof}
Assume that at least one hypothesis is false, since otherwise the result follows from Theorem \ref{thmcontrol}. 
Note that in the proof of Theorem \ref{thmcontrol}, we only used the independence of $(T_i:i\in \N)$ from $(T_i:i\in \N^c)$ once, namely in the step where we note that the probabilities   \eqref{eqsymmandindep1} and \eqref{eqsymmandindep2} are equal.
We used this independence to guarantee that the test statistics $(T_i:i\in \N)$ are independent from $A^+(t)$. In the present theorem, we do not assume independence of $(T_i:i\in \N)$ from $(T_i:i\in \N^c)$. 
However, consider the step where we showed that  the probabilities   \eqref{eqsymmandindep1} and \eqref{eqsymmandindep2} are equal. Here we will argue that they are equal with probability converging to 1.

We reason as follows.
We know that with probability converging to 1,  for every $t$ for which  $|\{j\in \N: T_j > \mu_j+t\}|>0$ or $|\{j\in \N: T_j < \mu_j-t\}|>0$, the following holds:
$$ \forall j\in \N^c: T_j>\delta_j +t.$$
Consider the event ${E}$, say, that that holds for all such $t$. Then for all $t$ for which  $|\{j\in \N: T_j > \mu_j+t\}|>0$ or $|\{j\in \N: T_j < \mu_j-t\}|>0$, we have $A^+(t)=|\N^c|$.
Consequently, under ${E}$,
\begin{equation}
\begin{aligned}
\pr\Big[&\sup\big\{t\geq 0: (1- \gamma) |\{j\in\N: T_j>\mu_j+t\}|  >   \gamma A^+(t) \big\}  \leq \\
&\sup\big\{t\geq 0:(1- \gamma)|\{j\in\N: T_j<\mu_j-t\}|  >   \gamma A^+(t) \big\}\Big]
\end{aligned}
\end{equation}
$$=$$
\begin{align*}
\pr\Big[&\sup\big\{t\geq 0: (1- \gamma) |\{j\in\N: T_j>\mu_j+t\}|  >   \gamma |\N^c| \big\}  \leq \\
&\sup\big\{t\geq 0:(1- \gamma)|\{j\in\N: T_j<\mu_j-t\}|  >   \gamma |\N^c| \big\}\Big].
\end{align*}
Since $\pr({E})$ converges to 1 as $n\rightarrow\infty$, the probability that the above holds converges to 1.
By Assumption \ref{asssym} the above  equals
\begin{align*}
\pr\Big[&\sup\big\{t\geq 0: (1- \gamma) |\{j\in\N: -T_j>-\mu_j+t\}|  >   \gamma |\N^c| \big\}  \leq \\
&\sup\big\{t\geq 0:(1- \gamma)|\{j\in\N: -T_j<-\mu_j-t\}|  >   \gamma |\N^c| \big\}\Big].
\end{align*}
With probability converging to 1 again, the above equals
\begin{equation}
\begin{aligned}
\pr\Big[&\sup\big\{t\geq 0: (1- \gamma) |\{j\in\N: -T_j>-\mu_j+t\}|  >   \gamma A^+(t) \big\}  \leq \\
&\sup\big\{t\geq 0:(1- \gamma)|\{j\in\N: -T_j<-\mu_j-t\}|  >   \gamma A^+(t) \big\}\Big].
\end{aligned}
\end{equation}
This, with probability converging to 1,  \eqref{eqsymmandindep1} and \eqref{eqsymmandindep2} are equal in the present setting.
Thus, we can obtain the same result as in Step 1 of the proof of Theorem \ref{thmcontrol}, except that it holds asymptotically:
$$\liminf_{n\rightarrow\infty}\{\pr(t_1\leq t_3) -\pr(t_3\leq t_1)    \}\geq 0$$
and hence 
$$\liminf_{n\rightarrow\infty}\pr(t_1\leq t_3) \geq 0.5.$$

By Step 2 of the  proof of Theorem \ref{thmcontrol}, we have $\pr(t_1\leq t_2|t_1\leq t_3)=1$ for every $n$, so that 
$\liminf_{n\rightarrow\infty}\pr(t_1\leq t_2)\geq 0.5.$ By Step 3 of the proof of Theorem \ref{thmcontrol}, we have $\pr(\forall t\geq s^+: FDP(t)\leq \gamma|t_1\leq t_2)=1$ for every $n$, so that it follows that $\liminf_{n\rightarrow\infty}\pr(\forall t\geq s^+: FDP(t)\leq \gamma)\geq 0.5$, as we wanted to show.

\end{proof}

\subsection{Proof of Theorem \ref{pvsgood}}

\begin{proof}
Suppose that for all $j\in \N$, $\mu_j=\delta_j$. Then the test statistics satisfy 
$$(T_j: j\in \N)   \,{\buildrel d \over =}\,  -(T_j: j\in \N)$$
and the $\pvs$ defined in the theorem then satisfy 
$$(P_j: j\in\N) \,{\buildrel d \over =}\,   (1-P_j: j\in\N).$$

However, we only know that $\mu_j\leq \delta_j$. But this means that $(T_j: j\in \N)$ is deterministically smaller than or equal to a vector $(T_1^*,...,T^*_N)$ satisfying
$$(T_1^*,...,T_N^*)   \,{\buildrel d \over =}\,  -(T_1^*,...,T_N^*). $$ 
Here, by ``smaller'', we mean that every entry of the former  is smaller than or equal to the corresponding entry of the latter.
Consequently, $(P_j: j\in\N)$ is  larger than or equal to a vector satisfying expression \eqref{sympvs}. 
By \citt{hemerik2024flexible}, if the $\pvs$ used satisfy expression \eqref{sympvs}, the methods in that paper are valid. Using larger $\pvs$ decreases $V(t)$ and increases the upper bounds from   \citt{hemerik2024flexible}. Consequently, the $\pvs$ $P_1,...,P_m$ can also validly be used in the methods from \citt{hemerik2024flexible}. This finishes the proof.
\end{proof}

\subsection{Proof of Theorem  \ref{pvsgoodet}}

\begin{proof}
For every $j\in \N$, define
$$
P_j' =
\begin{cases}
\int_{T_j - \mu_j }^{\infty } f_j(x) dx     & \text{if } \mu_j<0 \\
   \int_{-\infty}^{T_j-\mu_j} f_j(x) dx     & \text{if } \mu_j\geq0.
\end{cases}
$$ 
The numbers $P_j'$ are unknown in practice. However, since both $T_j-\mu_j$ and $f_j$ are symmetric about 0, the   $P_j'$ are jointly symmetric about 0.5.  For example, if $\mu_j<0$, we have 
$$P_j'=\int_{T_j - \mu_j }^{\infty } f_j(x) dx   \,{\buildrel d \over =}\,   \int_{-(T_j - \mu_j) }^{\infty } f_j(x) dx  = $$ 
$$   1 -   \int_{-\infty  }^{-(T_j - \mu_j) } f_j(x) dx  = 1- \int_{T_j - \mu_j }^{\infty } f_j(x) dx=1- P_j'$$
and more generally
$$  (P_j':j\in \N )       \,{\buildrel d \over =}\,       (1-P_j':j\in \N ).$$

Now, for every $j\in \N$ define
$$
P_j^* =
\begin{cases}
\int_{T_j +\delta_j }^{\infty } f_j(x) dx     & \text{if } \mu_j<0 \\
   \int_{-\infty}^{T_j-\delta_j} f_j(x) dx     & \text{if } \mu_j\geq0.
\end{cases}
$$ 
For every $j\in \N$ we have   $|\mu_j|\geq \delta_j$  and hence $T_j +\delta_j \leq T_j - \mu_j $ if $\mu_j<0$ and  $T_j - \delta_j \geq T_j -\mu_j $ if $\mu_j\geq 0$, so that $P_j^*\geq P_j'$.

Let $j\in \N$. We distinguish three cases. 

\emph{Case 1: $T_j$ and $\mu_j$ have the same sign.} Then we immediately have $P_j = P_j^*$. 

\emph{Case 2: $\mu_j\geq 0$ and $T_j<0$.} 
Then 
$$P_j - P_j^* =  \int_{T_j+\delta_j }^{\infty } f_j(x) dx    -   \int_{-\infty}^{T_j-\delta_j} f_j(x) dx.$$
The function $f_j$ was assumed to be symmetric about 0, so that the above equals
$$\int_{T_j+\delta_j }^{\infty } f_j(x) dx    -   \int_{-T_j+\delta_j  }^{ \infty} f_j(x) dx.$$
Since $T_j<-T_j$, the above is at least $0$, so that $P_j \geq P_j^*$.

\emph{Case 3: $\mu_j< 0$ and $T_j\geq 0$.} 
In this case, we likewise have
$$P_j - P_j^* =  \int_{-\infty}^{T_j-\delta_j} f_j(x) dx  -  \int_{T_j +\delta_j }^{\infty } f_j(x) dx    =$$
$$ \int_{-T_j+\delta_j }^{\infty} f_j(x) dx  -  \int_{T_j +\delta_j }^{\infty } f_j(x) dx \geq 0.$$

In conclusion, for every $j\in\N$, $P_j\geq P_j^*\geq  P_j'$ and  $(P_j':j\in \N )$  has the symmetry property \eqref{sympvs}.
By \citt{hemerik2024flexible}, that symmetry property  is sufficient for the methods in that paper to be valid. Using larger $\pvs$ decreases $V(t)$ and increases the upper bounds from   \citt{hemerik2024flexible}. Consequently, the $\pvs$ $P_1,...,P_m$ can also validly be used in the methods from \citt{hemerik2024flexible}.
\end{proof}

\subsection{Proof of claim from Remark  \ref{remTOST}}
Here we prove the claim from Remark \ref{remTOST} that the \emph{p}-value is Theorem \ref{pvsgoodet} is the TOST-\emph{p}-value.
\begin{proof}
The  general definition of the TOST-\emph{p}-value is in e.g. \citt{meyners2012equivalence} and \citt{lakens2017equivalence}.
In our setting, where the equivalence intervals $(\delta_j,\delta_j)$ are symmetric about 0, the following holds.
\begin{itemize}
\item If $T_j<0$, then the TOST-\emph{p}-value $P_j$ equals the right-sided \emph{p}-value for the null hypothesis that $\mu_j= -\delta_j$. 
\item Likewise, if $T_j\geq 0$, then the TOST-\emph{p}-value $P_j$ equals the left-sided \emph{p}-value for the null hypothesis that $\mu_j= \delta_j$. 
\end{itemize}
Thus,  if the observed statistic  $T_j^{\text{obs}}$ is smaller than 0, then to obtain the  TOST-\emph{p}-value, we compute  $\pr(T_j\geq T_j^{\text{obs}})$ under the assumption that $\mu_j =  -\delta_j$, which gives
\begin{align*}
& \pr(T_j\geq T_j^{\text{obs}}) = \pr(T_j-\mu_j\geq T_j^{\text{obs}}-\mu_j) = \\
&\int_{T_j^{\text{obs}}-\mu_j}^{-\infty} f_j(x)dx = \int_{T_j^{\text{obs}}+\delta_j}^{-\infty} f_j(x)dx.
\end{align*}
Likewise, if the observed statistic  $T_j^{\text{obs}}$ is larger than 0, then to obtain the TOST-\emph{p}-value, we compute
$\pr(T_j\leq T_j^{\text{obs}})$ under the assumption that $\mu_j =  \delta_j$, which gives
\begin{align*}
& \pr(T_j\leq T_j^{\text{obs}}) = \pr(T_j-\mu_j\leq T_j^{\text{obs}}-\mu_j) =\\
& \int_{-\infty}^{T_j^{\text{obs}}-\mu_j} f_j(x)dx= \int_{-\infty}^{T_j^{\text{obs}}-\delta_j} f_j(x)dx.
\end{align*}
This results in the expression \eqref{TOSTpv} from Theorem \ref{pvsgoodet}.
\end{proof}

\end{document}